\documentclass{article}

\usepackage{amsmath}
\usepackage{amsthm}
\usepackage{amsxtra}
\usepackage{amssymb}

\usepackage{geometry} 

\usepackage{graphicx}

\usepackage{balance}

\usepackage{authblk}

\usepackage{mathrsfs} 
\usepackage{eufrak} 

\usepackage{pdfpages}

\usepackage{verbatim} 
\usepackage{hyperref}
\usepackage[utf8]{inputenc}
\usepackage{microtype}

\usepackage{xspace}

\usepackage{algorithm}
\usepackage{algpseudocode}
\algtext*{EndWhile}
\algtext*{EndFor}
\algtext*{EndProcedure}
\algtext*{EndIf}
\algnewcommand\algorithmicinput{\textbf{Input:}}
\algnewcommand\Input{\item[\algorithmicinput]}
\algnewcommand\algorithmicoutput{\textbf{Output:}}
\algnewcommand\Output{\item[\algorithmicoutput]}

\graphicspath{{figs/}}

\newtheorem{definition}{Definition}
\newtheorem{lemma}[definition]{Lemma}
\newtheorem{proposition}[definition]{Proposition}

\newtheorem{theorem}[definition]{Theorem}

\renewcommand{\lg}{\log}
\newcommand{\comp}{\mathit{components}}
\newcommand{\cost}{\mathscr{C}}
\newcommand{\init}{V_{\operatorname{init}}}

\renewcommand{\root}{\mathfrak{r}}
\newcommand{\C}{\mathbb{C}}
\renewcommand{\S}{\mathcal{S}}
\newcommand{\U}{\mathcal{U}}
\newcommand{\T}{\mathcal{T}}
\newcommand{\hw}[2]{\#{#1}({#2})}
\newcommand{\bhw}[2]{\overline{\hw{#1}{#2}}}

\newcommand{\NP}{\textsf{NP}}
\renewcommand{\P}{\textsf{P}}

\newcommand{\union}{\operatorname{union}}
\newcommand{\find}{\operatorname{find}}

\DeclareRobustCommand{\ALG}{%
	\ifmmode
		\operatorname{ON}
	\else
		\text{ON}\xspace
	\fi
}
\DeclareRobustCommand{\OFF}{%
	\ifmmode
		\operatorname{OPT}
	\else
		\text{OPT}\xspace
	\fi
}
\DeclareRobustCommand{\APPROXALGO}{%
	\ifmmode
		\operatorname{APPROX}
	\else
		\text{APPROX}\xspace
	\fi
}

\title{Efficient Distributed Workload (Re-)Embedding\footnote{A version of this
paper will appear at SIGMETRICS'19. Authors are ordered alphabetically.}}

\author{Monika Henzinger}
\author{Stefan Neumann}
\author{Stefan Schmid}
\affil{University of Vienna,
			Faculty of Computer Science,
			Vienna, Austria}
\date{}

\begin{document}

\maketitle

\begin{abstract}
Modern networked systems
are increasingly reconfigurable,  
enabling \emph{demand-aware} infrastructures
whose resources can be adjusted 
according to the workload they currently serve.
Such dynamic adjustments can be exploited to 
improve
network utilization and hence performance, 
by \emph{moving} frequently
interacting communication partners closer, 
e.g., collocating them in the same server or datacenter.
However, dynamically changing the embedding of workloads 
is algorithmically challenging: communication patterns
are often not known ahead of time, but must be 
\emph{learned}. During the learning process,
overheads related to unnecessary moves 
(i.e., re-embeddings) should be minimized.
This paper studies a fundamental
model which captures the tradeoff between the benefits
and costs of dynamically collocating communication partners 
on $\ell$ servers, in an online manner. 
Our main contribution is a distributed online
algorithm which is asymptotically almost optimal, i.e.,
almost matches the lower bound 
(also derived in this paper)
on the competitive ratio of any (distributed or centralized)
online algorithm.
As an application, we show that our algorithm can be used to solve a distributed
union find problem in which the sets are stored across multiple servers.
\end{abstract}

\section{Introduction}
\label{sec:introduction}

Along with the trend towards more
\emph{data centric} applications 
(e.g., online services like web search,
social networking, financial services
as well as emerging applications such as distributed machine learning~\cite{survey2017datacenter,dist-ml}), 
comes a need to
\emph{scale out} such applications, and distribute
the workload across multiple servers or even datacenters.
However, while such parallel processing can improve
performance, it can entail a non-trivial load on the
interconnecting network.
 Indeed, distributed 
cloud applications, such as batch processing,
streaming, or scale-out databases, 
can generate a
significant amount of network
traffic~\cite{talk-about}.  

At the same time, emerging
networked systems are becoming increasingly flexible and
thereby provide novel 
opportunities to mitigate the overhead 
that distributed
applications impose on the network. In particular,
the more flexible and dynamic resource allocation 
(enabled, e.g.,
by virtualization) introduces a vision of
\emph{workload-aware} infrastructures which 
optimize themselves to the demand~\cite{ccr18san}. 
In such infrastructures, communication partners
which interact intensively, may
be \emph{moved} closer (e.g.,
collocated on the same 
server, rack, or datacenter) in an adaptive
manner, depending on the demand. 
This ``re-embedding'' of the workload allows to keep
communication local and reduce costs.
Indeed, empirical studies have shown that 
communication patterns in distributed applications
feature much locality,
which highlights the potential of such self-adjusting
networked systems~\cite{projector,Roy:2015,Benson-imc}. 

However, leveraging such resource reconfiguration
flexibilities 
to optimize performance, 
poses an algorithmic challenge. 
\emph{First,} 
while collocating
communication partners reduces communication cost,
it also introduces a \emph{reconfiguration cost} (e.g., due
to virtual machine migration).
Thus, an algorithm needs to strike a balance between
the benefits and the cost of such reconfigurations.
\emph{Second,} as workloads and communication patterns 
are usually not known ahead of time, 
reconfiguration decisions need to be made in
an \emph{online} manner, i.e., without knowing the future.
We are hence in the realm of online algorithms and
competitive analysis.

This paper studies the fundamental tradeoff underlying the optimization of such
workload-aware reconfigurable systems. In particular, we consider the design of
an online algorithm which, without prior knowledge of the workload, aims to
minimize communication cost by performing a small number of \emph{moves}
(i.e., migrations). In a nutshell
(more details will follow below), we consider a \emph{communication graph}
between $n$ vertices (e.g., virtual machines) which can be perfectly partitioned
among a set of $\ell$ servers (resp.~racks or datacenters) of a given
\emph{capacity}. We assume that the \emph{communication patterns}, which
partition the communication graph, consist of $n/\ell$ vertices and that once
the whole communication graph was revealed, each server must contain exactly one
communication pattern.

The communication graph is initially unknown and revealed to the algorithm in an
online manner, edge-by-edge, by an adversary who aims to maximize the cost of
the given algorithm. The cost here consists of \emph{communication cost} and
\emph{moving cost}: The algorithm incurs one unit cost if the two endpoints
(i.e., communication partners) of the request belong to different servers. After
each request, the algorithm can reconfigure the infrastructure and move
communication endpoints from one server to another, essentially
\emph{repartitioning} the communication partners; however, each move
incurs a cost of $\alpha>1$.

In other words, this paper considers the problem of \emph{learning a partition},
i.e., an optimal assignment of communication partners to servers, at low
communication and moving
cost.  Interestingly, while the problem is natural and fundamental,
 not much is
known today about the algorithmic challenges underlying this problem,
except for the negative result that no good competitive algorithm can exist  
if communication partners can change arbitrarily over time~\cite{obr-original}.
This lower bound motivates us, in this paper, to focus on the online 
\emph{learning variant} where the communication partners are unknown but fixed.
At the
same time, as we will show, the problem features interesting connections to
several classic problems. Specifically, the problem can be seen as a
\emph{distributed} version of classic online caching
problems~\cite{competitive-analysis} or an \emph{online} version of the $k$-way
partitioning problem~\cite{ethan18optimal}.

\subsection{Our Contributions}

We initiate the study of a fundamental problem, how to learn and re-embed
workload in an online manner, with few moves.  We make the following main
contributions.

We present a distributed $O((\ell \lg \ell \lg n) / \varepsilon)$-competitive
online algorithm for servers of capacity $(1+\varepsilon)n/\ell$, where
$\varepsilon \in (0,1/2)$. We allow the servers to have $\varepsilon n/\ell$
more space than is strictly needed to embed its corresponding communication
pattern (which is of size $n/\ell$); we denote this additional space as
\emph{augmentation}.  Such augmentation is also needed, as our lower bounds
discussed next show.

We show that there are inherent limitations of what online algorithms can
achieve in our model:
We derive a lower bound of $\Omega(1/\varepsilon + \lg n)$ on the competitive
ratio of any deterministic online algorithm given servers of capacity at least
$(1+\varepsilon)n/\ell$.  This lower bound has several consequences:
(1) To obtain $O(\lg n)$-competitive algorithms, the servers must have
$\Omega(n/(\ell \lg n))$ augmentation. 
(2)~If the servers have $\Omega(n/\ell)$ augmentation (e.g., each server has
	10\% more capacity than the size of its communication pattern), our
	algorithm is optimal up to an $O(\ell \lg \ell)$ factor.
Thus, our results are particularly interesting for large servers, e.g., in a
wide-area networking context where there is usually only a small number of
datacenters where communication partners can be collocated (e.g., $\ell = 20$):
if each datacenter (``server'') has augmentation $0.1 \cdot n/\ell$ , our
algorithm is optimal up to constant factors.

The distributed algorithms we present not only provide good competitive ratios
but they are also highly efficient w.r.t.\ the network traffic they cause. In
fact, we show that for $\ell = O(\sqrt{\varepsilon n})$ servers, running the
algorithms introduces only little overhead in network traffic and that this
overhead is asymptotically negligible (see Section~\ref{sec:distributed}).

While the previous algorithms require exponential time, we also present
polynomial time algorithms at the cost of a slightly worse competitive ratio of
$O((\ell^2 \lg n \lg \ell) / \varepsilon^2)$ in Section~\ref{sec:fast}.

As a sample application of our newly introduced model we present a distributed
union find data structure~\cite{galler64improved,tarjan84worst} (also known as
disjoint-set data structure or merge-find data structure) in
Section~\ref{sec:union-find}:
There are $n$ items from a universe which are distributed over $\ell$
servers; each server can store at most $(1+\varepsilon) n/\ell$ items and each
item belongs to a unique set. The operation \emph{union} allows to merge two
sets. In our setting, we require that items from the same set must be assigned
to the same server. To reduce the network traffic, our goal is to minimize the
number of item moves during union operations. For example, when two sets are
merged which are assigned to different servers, then the items of one of the
sets must be reassigned to another server.  We compare against an optimal
offline algorithm which knows the initial assignment of all items and all union
operations in advance. We obtain the same competitive ratios as above. We
believe that this distributed union find data structure will be useful as a
subroutine for several problems such as merging duplicate websites in search
engines~\cite{broder97syntactic}.

We also show that our algorithms solve an online version of the $k$-way
partition problem in Section~\ref{sec:partition}.

\subsection{Organization}

We introduce our model
formally in Section~\ref{sec:preliminaries}.
To ease the readability, we first explore centralized online algorithms 
that efficiently collocate communication patterns 
for $\ell=2$ servers in Section~\ref{sec:two-servers}, and then study the general case
of $\ell > 2$ servers in Section~\ref{sec:l-servers}.
In Section~\ref{sec:distributed-fast} we show how the previously derived centralized
algorithms can be made distributed and how the algorithm can be implemented in
polynomial time at the cost of a slightly worse competitive ratio.
We provide the lower bounds in Section~\ref{sec:lowerbounds}.
Section~\ref{sec:applications} provides a distributed union find data structure
and a result for online $k$-way partitioning; these problems serve as sample
applications of the problem we study.  After reviewing related work in
Section~\ref{sec:related}, we conclude our contribution in
Section~\ref{sec:conclusion}.

\section{Model}
\label{sec:preliminaries}

We start by formally introducing the model which we will be studying in this paper.
We consider a set of vertices $V$ (e.g., a set of virtual machines) which
interact according to an initially unknown communication pattern, which can be
represented as a communication graph $G = (V,E)$ with $n=|V|$ vertices and
$m=|E|$ edges. The vertices of $G$ are partitioned into $\ell$ sets $V_0, \dots,
V_{\ell-1}$ where each $V_i$, forming a connected communication component (the
workload), has size\footnote{Note that in general $n/\ell$ is not always an
	integer and we would have to take rounding into account.  However, we ignore
	this technicality for better readability of the paper.}
$n/\ell$; the connected components of $G$ coincide with the sets $V_i$.  The
sets $V_i$ are the communication patterns which need to be recovered by the
online algorithm, henceforth called \emph{ground truth components}.

The communicating vertices $V$ need to be assigned
to $\ell$ \emph{servers} $S_0, \dots, S_{\ell-1}$. 
Accordingly, we define an \emph{assignment} (the embedding) which is a
function from the vertices to the servers.  
The \emph{load} of a server $S_j$ is
the number of vertices that are assigned to it.  
An assignment is
\emph{valid} if each server has load at most $n/\ell + K$ and 
we call $n/\ell + K$ the
\emph{capacity} of the servers and
$K$ the \emph{augmentation}. If $K=0$, the total server capacity
exactly matches the number of vertices.
The \emph{available capacity} of a server is the 
difference between the
server's capacity and its load.
An assignment is \emph{perfectly
balanced} if each server has load exactly $n/\ell$.  
We assume that when the
algorithm starts, we have a perfectly balanced assignment. 
We will write $V(S_j)$ to denote 
the set of vertices assigned to server
$S_j$ and $\init(S_j)$ for the set of vertices \emph{initially} assigned to
server $S_j$. We say that an assignment is a \emph{perfect partitioning} if it
satisfies $\{V(S_0), \dots, V(S_{\ell-1})\} = \{V_0, \dots, V_{\ell-1}\}$, 
i.e., the
vertices on the servers coincide with the 
connected components of~$G$.

The communication graph $G = (V,E)$ 
is revealed by an adversary in an 
\emph{online manner}, as a sequence of edges 
$\sigma = (e_1, \dots, e_r)$, 
where $r$ denotes the number of communication requests and $e_i \in E$ for each $i$. Note that 
the adversary can only
provide edges which are present 
in $E$ and that each edge can appear multiple
times in the sequence of edges. We assume 
that the sequence of the edges
provided by the adversary reveals the 
ground truth components $V_i$, i.e., after
having seen all edges in $\sigma$ the algorithm can compute the
connected components of $G$ which (by assumption) coincide with the ground truth components $V_i$.
We present an illustration of the model in Figure~\ref{fig:model}.

\begin{figure}
  \begin{centering}
    \includegraphics[width=.99\columnwidth]{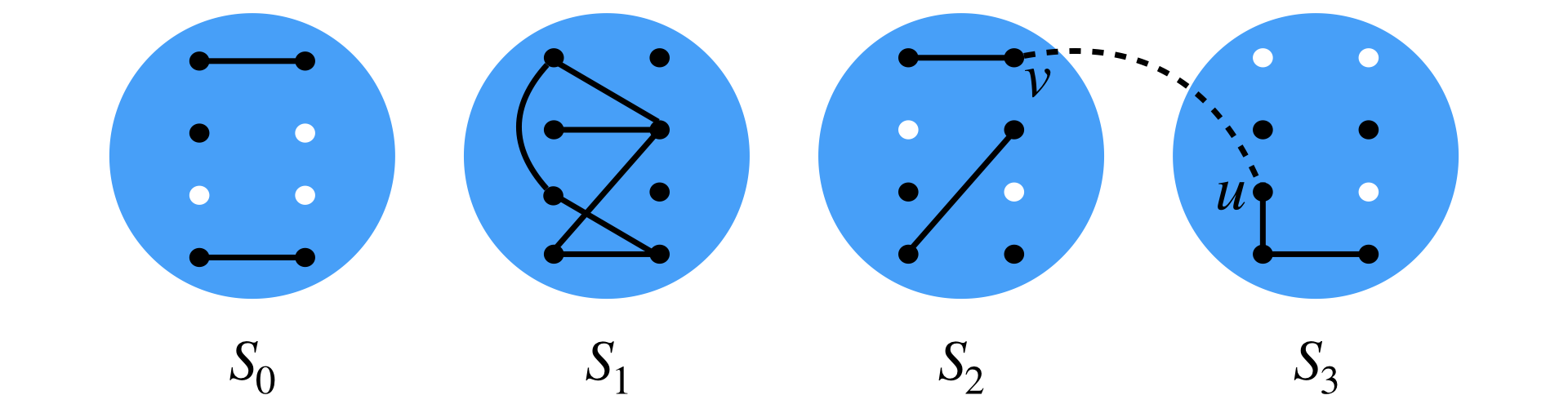}
	\caption{An illustration of the model we consider. In the picture there are
		$\ell = 4$ servers each depicted by a blue circle. Vertices assigned to
		a server are represented by black dots whereas white dots represent
		unused server capacities. Note that there are $n = 24$ vertices and each
		server has capacity $(1+\varepsilon)n/\ell = 8$ for $\varepsilon = 1/3$.
		In the picture, server $S_0$ has load $5$ and server $S_1$ has load $8$.
		When two vertices communicated, we draw an edge between them depicted by
		a black line. Observe how this naturally gives rise to connected
		components and note that $S_1$ contains a ground truth component of size
		$n/\ell = 6$. If the adversary were to insert the edge $(u,v)$ next, the
		algorithm could, for example, move the connected component containing
		$v$ to $S_3$ at cost $2\alpha$.}
	\label{fig:model}
  \end{centering}
\end{figure}

Now an online algorithm must 
iteratively change the assignment such that
eventually the assignment is a perfect partitioning.

The 
reassignment needs 
to be done while minimizing certain
communication and migration cost.  
If an edge $e = (u,v)$ provided by the
adversary has both endpoints in the same server $S_i$
at the time of the request, 
an algorithm incurs no costs. If
$u$ and $v$ are in different servers $S_i$ and $S_j$, 
then their communication
cost is $1$. Reassigning, i.e., \emph{moving},
a vertex $u$ from a server
$S_i$ to a server $S_j$ costs $\alpha > 1$.

When measuring the cost of an 
online algorithm, we will compare against an optimal 
offline algorithm denoted by \OFF.
\OFF has \emph{a priori} knowledge of
the communication graph $G = (V,E)$
as well as the 
given the sequence of all edges 
$\sigma=(e_1, \dots, e_r)$.
In other
words, \OFF can compute the assignment of
vertices to servers  
which provides the minimum
migration cost from the initial assignment.

Now let the cost paid by an online algorithm 
be denoted by \ALG and let the cost
of the optimal offline algorithm be denoted by
\OFF.
We consider the design of an online 
algorithm $\ALG$ 
which minimizes the (strict)
\emph{competitive ratio} defined as
 $\frac{\ALG}{\OFF}$. 

\paragraph{The Role of Connected Components}
We will briefly discuss how connected components are induced by subsequence of
$\sigma$ and how we will treat connected components in our algorithms. We then
give a reduction which helps us to avoid considering communication costs in our
proofs.

Recall that the adversary provides a sequence of edges $\sigma$ to an algorithm
in an online manner. As this happens, an algorithm can keep track of all edges
it has seen so far. Let this set of edges be $E'$.  Using the edges in $E'$, the
algorithm can compute the connected components $C_1, \dots, C_q$ which are
induced by $E'$.  Here, $q$ denotes the current number of connected components.

To obtain a better understanding of the relationship between the connected
components $C_i$ and the ground truth components $V_j$, we make four
observations:
(1)~When the algorithm starts, all connected components $C_i = \{ v_i \}$ only
consist of single vertices (because $\sigma$ has not yet revealed any edges).
(2)~When a previously unknown edge $e = (u,v)$ is revealed which has its
endpoints in different connected components $C_u$ and $C_v$, these connected
components get merged. 
(3)~Suppose a subsequence of $\sigma$ induces $q > \ell$ connected components
$C_i$ (i.e., $\sigma$ has not yet revealed the whole graph~$G$). Then for each
ground truth component $V_j$ there exists a subset $\C \subset
\{C_1,\dots,C_q\}$ of the connected components such that $V_j = \bigcup_{C \in \C} C$.
(4)~When an algorithm terminates (and, hence, $\sigma$
revealed all edges in $E$), there exists a one-to-one correspondence between the
connected components $C_i$ and the ground truth components $V_j$. 

By assumption on the input from the adversary, when all of $\sigma$ was
revealed, $E'$ reveals the ground truth components $V_0, \dots, V_{\ell-1}$.  Thus, in
total there will be exactly $n - \ell$ edges connecting vertices from different
connected components.

All of the algorithms we consider in this paper have the property that they
always assign vertices of the same connected component to the same server.
This property implies that the communication cost paid by such an algorithm is
bounded by its moving cost (we prove this in the following lemma). Hence, in the
rest of the paper we only need to bound the moving costs of our algorithms to
obtain a bound on their total costs.

\begin{lemma}
\label{lem:reduction}
	Suppose an algorithm $\mathcal{A}$ always assigns all vertices of the same
	connected component to the same server and pays $\cost$ for moving vertices.
	Then its communication cost is at most $\cost$. Furthermore, its total cost is
	at most $2\cost$.
\end{lemma}
\begin{proof}
	Suppose the adversary provides an edge $(u,v)$. We consider two cases.
	\emph{Case~1:} $u$ and $v$ are assigned to the same server. Then
	$\mathcal{A}$ does not pay any communication costs. \emph{Case~2:} $u$~and
	$v$ are
	assigned to connected components $C_u$ and $C_v$ on different servers. Then the
	algorithm needs to pay $1$ communication cost.  However, in this case
	$\mathcal{A}$ must move $C_u$ or $C_v$ to a different server at the cost of at
	least $\alpha > 1$.  Hence, the moving cost is larger than the communication
	cost. We conclude that $\mathcal{A}$'s total communication cost is at most
	$\cost$.  By summing the two quantities, we obtain the second claim of the
	lemma.
\end{proof}

While in Lemma~\ref{lem:reduction} we have shown that algorithms which always
collocate connected components immediately are efficient w.r.t.\ their total
cost, in Section~\ref{sec:connected-components} we show that any efficient
algorithm must satisfy a similar (slightly more general) property.

Throughout the rest of the paper, we write $|C|$ to denote the number of
vertices in a connected component $C$. For a vertex $u$, we write $C_u$ to
denote the connected component $C$ which contains $u$.

\section{Online Partition for Two Servers}
\label{sec:two-servers}

In this section, we consider the problem of learning
a communication graph with few moves
with \emph{two} servers. 
As we will see later, the concepts introduced in this
section will be useful when solving the problem
with $\ell > 2$ servers.  We derive the following result.

\begin{theorem}
\label{thm:2servers}
	Consider the setting with two servers of capacity $(1+\varepsilon)n/2$ for
	$\varepsilon \in (0,1)$, i.e., the augmentation is $\varepsilon n/2$. Then
	there exists an algorithm with competitive ratio $O( (\lg n) / \varepsilon)$.
\end{theorem}

The proof is organized as follows. 
We first characterize
the optimal solution by \OFF 
in Section~\ref{sec:cost-opt}. 
We then present an
algorithm which is efficient whenever 
\OFF incurs ``significant cost'', 
in Section~\ref{sec:small-large-rebalance}. 
In Section~\ref{sec:majority-voting},
we describe an
algorithm which is efficient whenever the solution by 
\OFF is ``cheap''.
We prove
Theorem~\ref{thm:2servers} via a combination of the 
two algorithms in Section~\ref{sec:proof-thm-2servers}.

\subsection{Costs of \OFF}
\label{sec:cost-opt}

The following lemma gives a precise characterization of the cost paid by \OFF in
the two server case. It introduces a parameter $\Delta$ which equals the number
of vertices moved by \OFF and which we will be using throughout the rest of this
section.

\begin{lemma}
\label{lem:cost-opt-2server}
	Suppose $\ell = 2$ and the vertices initially assigned 
	to the servers $S_i$
	are given by the sets $\init(S_i)$ for $i = 0,1$. 
	Then the cost of \OFF is
	$2 \alpha \Delta$, where
	\begin{align*}
		\Delta = \min\{ |\init(S_0) \cap V_0|, |\init(S_0) \cap V_1| \}.
	\end{align*}
	It follows immediately that $\Delta \leq n/4$ (as $|\init(S_0)| = n/2$).
\end{lemma}
\begin{proof}
	Recall that our model forces \OFF to provide a final assignment satisfying
	$\{V(S_0), V(S_1)\} = \{V_0, V_1\}$, i.e., \OFF must produce a final assignment
	which coincides with the ground truth components (even if paying for each
	communication request individually and not relocating any vertices might be
	cheaper). Thus, we can assume that \OFF performs all vertex moves in the
	beginning, to avoid paying any communication cost.  Since the edge sequence
	$\sigma = (e_1, \dots, e_r)$ provided by the adversary is assumed to reveal
	the connected components $V_0$ and $V_1$, \OFF can compute $V_0$ and $V_1$
	before it performs any moves.

	As there are only two servers, one of them must contain at least half of
	the vertices from $V_0$ in the initial assignment. Now 
	let us first assume that this
	server is $S_0$; this setting is illustrated in
	Figure~\ref{fig:cost-opt-2server}. In this case, \OFF can move the $\Delta$
	vertices in $\init(S_0) \cap V_1$ to $S_1$ and those in $\init(S_1) \cap V_0$ to
	$S_0$. It is easy to verify that this yields an assignment satisfying
	$\{V(S_0), V(S_1)\} = \{V_0, V_1\}$ and that the moving cost is minimized.
	Further, the cost for this reassignment is exactly $2 \alpha \Delta$.

	The second case where $S_1$ contains more than half of the vertices from
	$V_0$ in the initial assignment is symmetric.
\end{proof}

While in Lemma~\ref{lem:cost-opt-2server} we have presented the lower bound
w.r.t.\ server $S_0$, we could also express the lower bound in terms of server
$S_1$. We then obtain the following equality:
\begin{align*}
	\Delta = \max_{i=0,1} \min_{j=0,1} | \init(S_i) \cap V_j |.
\end{align*}

\begin{figure}
  \begin{centering}
    \includegraphics[width=.99\columnwidth]{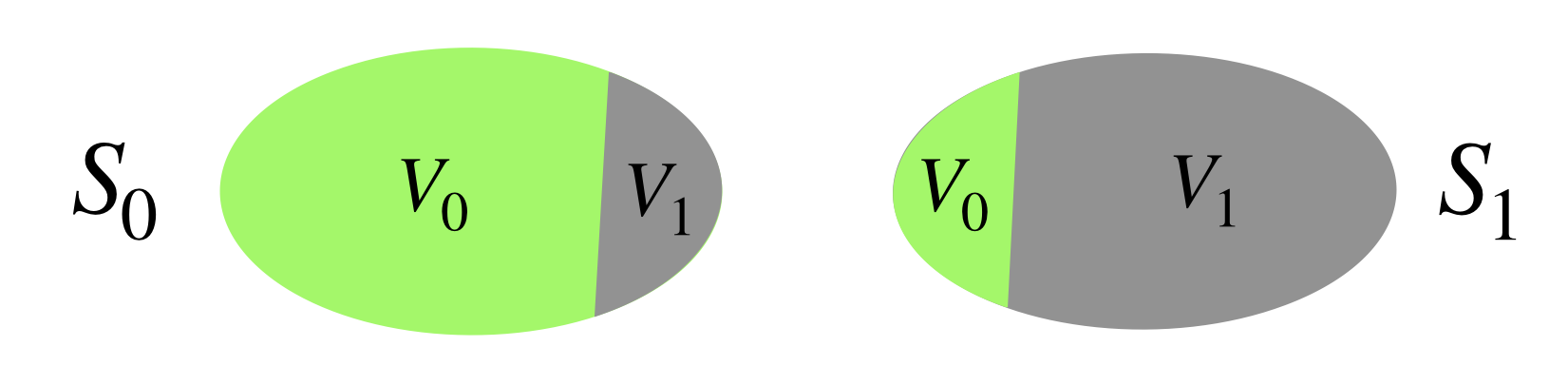}
	\caption{The initial assignment considered in the proof of
		Lemma~\ref{lem:cost-opt-2server}. The green and grey areas of the
		servers correspond to subsets of $V_0$ and $V_1$. Server $S_0$ ($S_1$)
		contains most of the vertices from $V_0$ ($V_1$). Here, \OFF would move
		the green part from $S_1$ to $S_0$ and the grey part from $S_0$ to
		$S_1$.}
	\label{fig:cost-opt-2server}
  \end{centering}
\end{figure}

\subsection{The Small--Large--Rebalance Algorithm}
\label{sec:small-large-rebalance}

A natural idea to obtain a small number of vertex moves is to proceed as
follows. Whenever two vertices $u$ and $v$ belonging to different
connected components communicate, the algorithm merges 
their connected components. 
If the two components were already assigned 
to the same server, no vertex moves are required. 
If $u$ and $v$ are
assigned to different servers, we move the smaller connected 
component to the
server of the larger connected component. This algorithm 
is efficient in that it
never performs more than $O(n \lg n)$ vertex moves (see
		Lemma~\ref{lem:small-large-only}).

However, the algorithm could require much augmentation, as it
does not account for server capacities. 
Thus, we propose the following extension called the
\emph{Small--Large--Rebalance Algorithm}: Whenever a server exceeds its
capacity, the algorithm computes a perfectly balanced assignment of the
vertices which respects the previously observed connected components; we call
this a \emph{rebalancing step}.  We provide pseudocode in
Algorithm~\ref{algo:small-large-rebalance}.

Section~\ref{sec:rebalancing:two-servers} shows how such a rebalancing step can
be implemented in $O(n^2)$ time. Later, we show that there can be at most
$O((\lg n)/\varepsilon)$ such rebalancing steps which implies that the total
running time Algorithm~\ref{algo:small-large-rebalance} is $O((n^2 \lg n)/\varepsilon)$.

Note that Algorithm~\ref{algo:small-large-rebalance} also works in the setting
with $\ell$ servers for $\ell > 2$. We will analyze this more general algorithm in
Section~\ref{sec:small-large-rebalance-lservers}.

\subsubsection{Analysis}
To analyze Algorithm~\ref{algo:small-large-rebalance}, we first consider the
algorithm from the first paragraph which does not have the rebalancing step.
When the algorithm moves a smaller component to the server of a larger
component, we call this a \emph{small-to-large step}.

\begin{lemma}
\label{lem:small-large-only}
	Consider the algorithm which always moves the smaller connected component to
	the server of the larger connected component when it obtains an edge
	between vertices from different connected components. The algorithm moves
	each vertex at most $O(\lg n)$ times. Its total number of vertex moves is
	$O(n \lg n)$.
\end{lemma}
\begin{proof}
	Consider any vertex $v$. 
	We use the following accounting:
	Whenever $v$ is in the a smaller component that is
	moved, add a token to $v$. Now observe that whenever $v$ gains a token, the
	size of its component at least doubles. This implies that $v$ can be in the
	smaller component only $O(\lg n)$ times. Thus, $v$ cannot accumulate more
	than $O(\lg n)$ tokens.  Since this holds for each of the $n$ vertices, the
	total number of moves is $O(n \lg n)$.
\end{proof}

The following lemma provides the analysis for
Algorithm~\ref{algo:small-large-rebalance} which performs small-to-large steps
and rebalancing steps.

\begin{lemma}
\label{lem:small-large-rebalance}
	Suppose both servers have capacity $(1+\varepsilon)n/2$, i.e., the
	augmentation is $\varepsilon n/2$ for $\varepsilon \in (0,1)$.
	Then Algorithm~\ref{algo:small-large-rebalance} performs
	$O((\lg n) / \varepsilon)$ rebalancing steps and
	$O((n \lg n) / \varepsilon)$ vertex moves.
\end{lemma}
\begin{proof}
	We prove the bound on the number of vertex moves; the claim about the number
	of rebalancing steps is proved along the way.  Note that all vertex moves
	performed by the algorithm originate from either small-to-large steps or
	from rebalancing steps. We bound the number of each of these vertex moves
	separately.

	Note that the token-based argument from Lemma~\ref{lem:small-large-only}
	still applies to the small-to-large steps of
	Algorithm~\ref{algo:small-large-rebalance}. This implies that the total
	number of vertex moves due small-to-large steps is $O(n \lg n)$.

	Now consider the vertex moves caused by the rebalancing steps and recall
	that the initial assignment is perfectly balanced. Whenever a server
	exceeds its load, the small-to-large steps of the algorithm must have moved
	at least $\varepsilon n/2$ vertices (because the augmentation of one of the
	servers is exceeded).  This can only happen
	$O((n \lg n) / (\varepsilon n)) = O((\lg n) / \varepsilon)$ times since the
	total number of vertex moves due to small-to-large steps is $O(n \lg n)$.
	Hence, the number of rebalancing steps is at most $O((\lg n) / \varepsilon)$.
	Since each rebalancing step performs $O(n)$ vertex moves, the lemma follows.
\end{proof}

\begin{algorithm}[t!]
\caption{The Small--Large--Rebalance Algorithm}\label{algo:small-large-rebalance}
\begin{algorithmic}[1]
\Input{A sequence of edges $\sigma=(e_1, \dots, e_r)$}
\Procedure{SmallLargeRebalance}{$e_1, \dots, e_r$}
	\For{$i = 1, \dots, r$}
		\State $(u, v) \leftarrow e_i$
		\If{$C_u$ and $C_v$ are not assigned to the same server}
			\State \Comment{We must move $C_u$ and $C_v$ to the same server.}
			\State Assume w.l.o.g.\ that $|C_u| \leq |C_v|$
			\If{the server of $C_v$ has available capacity $|C_u|$}
				\State \label{line:small-large-rebalance:small-to-large}
			   	Move $C_u$ to the server of $C_v$
				\State \Comment{Small-to-large step}
			\Else			\Comment{Rebalancing step}
				\State \label{line:rebalance}
				Move to a perfectly balanced assignment respecting the connected
				components
			\EndIf
		\EndIf
		\State Merge $C_u$ and $C_v$
	\EndFor
\EndProcedure
\end{algorithmic}
\end{algorithm}

\subsubsection{More Efficient Rebalancing}
\label{sec:efficient-small-large-rebalance}

We next propose a
better rebalancing strategy
which makes Algorithm~\ref{algo:small-large-rebalance} 
more efficient. So far, we used $\Theta(n)$ vertex moves for each
rebalancing operation at the cost of $\Theta(\alpha n)$. We now bring the
rebalancing cost down to $O(\OFF)$.

We adjust Algorithm~\ref{algo:small-large-rebalance} in the following way:
Instead of rebalancing by taking \emph{any} perfectly balanced assignment
respecting the connected components (Line~\ref{line:rebalance}), we choose a
perfectly balanced assignment respecting the connected components \emph{which
minimizes the number of vertex moves from the initial solution}. We call such an
assignment \emph{cheap}.

To find a cheap assignment, the algorithm could simply do the following:
(1)~Recall the initial assignment. (2)~Exhaustively enumerate all perfectly
balanced assignments respecting the connected components. (3)~Among all of these
assignments find one which is cheap.  While such a simple algorithm can in
principle be computationally costly, we can here exploit the online model of
computation which allows us unlimited computational power.  In
Section~\ref{sec:fast} we show how less efficient rebalancing strategies can be
implemented in polynomial time and we obtain slightly worse competitive ratios.

With the improved rebalancing strategy, we obtain
Proposition~\ref{prop:small-large-rebalance}.

\begin{proposition}
\label{prop:small-large-rebalance}
	Suppose all servers have capacity $(1+\varepsilon)n/2$, $\varepsilon > 0$.
	Then the number of vertex reassignments performed by
	Algorithm~\ref{algo:small-large-rebalance} with more efficient rebalancing
	is $O(n \lg n + (\Delta \lg n) / \varepsilon)$, where $\Delta$ is the number
	of vertex moves used by \OFF.
\end{proposition}
\begin{proof}
	First, note that the number of vertex moves for moving smaller components to
	larger components (Line~\ref{line:small-large-rebalance:small-to-large}) is
	$O(n \lg n)$, by exactly the same arguments used in the proof of
	Lemma~\ref{lem:small-large-rebalance}.

	Second, we bound the number of vertex moves required for the rebalancing
	operations. Whenever the algorithm needs to rebalance, we can assume (for
	the sake of the analysis) that the algorithm makes the following three
	steps:
	(1)~Roll back all changes done by small-to-large moves
	(Line~\ref{line:small-large-rebalance:small-to-large}) \emph{since the last
	rebalancing operation}. Thus, after rolling back we have the same assignment
	as after the last rebalancing operation. 
	(2)~Roll back to the initial assignment (by undoing the last rebalancing
	operation).
	(3)~Move to a cheap assignment.

	Observe that Step~(1) and~(2) of the previous three step procedure increase
	the number of vertex moves only by a constant factor compared to when the
	algorithm does not roll back: 
	In total, Step~(1) only adds additional $O(n \lg n)$ vertex moves because
	each small-to-large move is undone exactly once. Step~(2) only doubles the
	number of vertex moves for moving to cheap assignments as each rebalancing
	is only undone once.

	Thus, we can complete the proof 
	if we can show that the total number of vertex
	moves for moving from the initial assignment to the cheap assignments is
	bounded by
	$O((\Delta \lg n) / \varepsilon)$.

	By Lemma~\ref{lem:small-large-rebalance}, the number of rebalancing steps is
	bounded by $O((\lg n) / \varepsilon)$.  Now we argue that for moving from the
	initial solution to each cheap assignment, the rebalancing moves at most
	$O(\Delta)$ vertices: Every time the algorithm computes a cheap rebalancing,
	the final solution obtained by \OFF is a perfectly balanced assignment
	respecting the connected components. Thus, the number of vertex moves to
	obtain a cheap rebalancing is bounded by the number of moves performed by
	\OFF which is $O(\Delta)$. This finishes the proof.
\end{proof}

\subsection{The Majority Voting Algorithm}
\label{sec:majority-voting}

We now present an algorithm which works well whenever the cost paid by \OFF is
small, i.e., when \OFF only needs to move few vertices.  The issue with
Algorithm~\ref{algo:small-large-rebalance} from
Section~\ref{sec:small-large-rebalance} is that
during its execution, it might deviate much from the initial assignment (and
thus move many vertices). The following algorithm has the property that it
always stays close to the initial assignment.

For ease of readability, we will often refer to the two servers as the
\emph{left} and \emph{right} servers, respectively, 
instead of calling them $S_0$ and $S_1$.

Our algorithm starts by coloring vertices on the left server yellow and on the
right server black. Throughout the execution of the algorithm, the vertices
will keep this initially assigned color. The algorithm then follows the idea of
always moving the smaller connected component to the server of the larger
connected component; we will refer to this as \emph{small-to-large step}. To
stay close to the initial assignment, whenever the number of vertices in a newly
merged connected component surpasses a power of $2$, the algorithm performs a
majority vote and moves the component to the server where more of its vertices
originate from.  More formally, we say that a set of vertices (e.g., a connected
component) has a \emph{yellow (black) majority} if it contains more yellow
(black) vertices than black (yellow) vertices. In the majority voting step, the
algorithm moves a component with a yellow (black) majority which is currently on
the right (left) server to the left (right) server. The pseudocode for this
procedure is stated in Algorithm~\ref{algo:majority-voting}.\footnote{Note that in Algorithm~\ref{algo:majority-voting} the following is
	possible when a component $C_u$ is merged with a component $C_v$: $C_u$ is
	moved from $S$ to $S_v$ due to a small-to-large step and immediately after
	that $C_u \cup C_v$ is moved back to $S$ due to a majority-voting step.
	Thus, it would be more efficient to compute the result of the
	majority-voting step earlier and to move $C_v$ to $S$ immediately (without
	ever moving $C_u$ to $S_v$). This modification would be slightly more
	efficient but it would affect the competitive ratio of the algorithm only by
	at most a constant factor. Thus, to simplify our analysis, we ignore this
	modification.
}

\begin{algorithm}[t!]
\caption{The Majority Voting Algorithm}\label{algo:majority-voting}
\begin{algorithmic}[1]
\Input{A sequence of edges $\sigma=(e_1, \dots, e_r)$}
\Procedure{MajorityVoting}{$e_1, \dots, e_r$}
	\State Color all vertices assigned to the left server yellow and all
	vertices assigned to the right server black
	\For{$i = 1, \dots, r$}
		\State $(u,v) \leftarrow e_i$
		\State Suppose w.l.o.g.\ that $|C_u| \leq |C_v|$
		\If{$C_u$ and $C_v$ are on different servers}
			\State\label{line:small-to-large} Move $C_u$ to the server of $C_v$
									\Comment{Small-to-large step}
		\EndIf
		\State Merge $C_u$ and $C_v$
		\If{there exists an $i \in \mathbb{N}$ s.t.\ $|C_u| < 2^i$, $|C_v| <
			2^i$ and $|C_u \cup C_v| \geq 2^i$}		 \Comment{Majority voting step}
		\label{line:majority-vote}
			\If{$C_u \cup C_v$ has a yellow majority}
				\State Move $C_u \cup C_v$ to the left server
			\EndIf
			\If{$C_u \cup C_v$ has a black majority}
				\State Move $C_u \cup C_v$ to the right server
			\EndIf
		\EndIf
	\EndFor
\EndProcedure
\end{algorithmic}
\end{algorithm}

The reason for introducing the majority voting step is that it keeps the
assignments produced by the algorithm during its runtime close to the initial
assignment. Due to this property, we can show that the cost of
Algorithm~\ref{algo:majority-voting} is always close to the cost of \OFF. The
formal guarantees are stated in Proposition~\ref{prop:majority-voting}.

\begin{proposition}
\label{prop:majority-voting}
	Let $\Delta$ be the number of vertex moves performed by \OFF (see
	Section~\ref{sec:cost-opt}).  Then Algorithm~\ref{algo:majority-voting} is
	$O(\lg n)$-competitive and the load of both servers is bounded by $n/2 +
	4\Delta$.
\end{proposition}

We devote rest of this subsection to the proof of the proposition. We start
bounding the augmentation.  For the proofs recall that $V_0$ and $V_1$ are the
ground truth connected components of~$G$.

In the following we are interested in what happened to a connected component
since its last majority vote.  To this end, we decompose it into a sequence of
smaller connected components such that first a majority vote is performed and after
that, only small-to-large steps are performed. For all of these small-to-large
steps, the component will stay on the server that was picked by the majority
vote.  The following definition makes this notion formal.

\begin{definition}[Doubling Decomposition]
\label{def:double-decomposition}
	Let $C$ be a connected component and let $s \in \mathbb{N}$ be such that
	$2^t \leq |C| < 2^{t+1}$. Consider $k$ disjoint sets of vertices $C_i
	\subseteq V$ and let $\C_j = \bigcup_{i=1}^j C_i$ for $j = 1,\dots,k$.

	A sequence $(C_1,\dots,C_k)$ is a \emph{doubling decomposition} of $C$ if
	the following properties hold:
	\begin{enumerate}
		\item \label{def:doubling-decomposition:union}
			$C = \C_k = \bigcup_{i = 1}^k C_i$,
		\item \label{def:doubling-decomposition:merges}
			during the execution of the algorithm, first $\C_1 \cup C_2$
			are merged, then $\C_2 \cup C_3$ are merged, and, more generally,
			$\C_{i-1} \cup C_i$ is merged before $\C_i \cup C_{i+1}$,
		\item \label{def:doubling-decomposition:move}
			for each $i = 1,\dots,k-1$, $|\C_i| \geq C_{i+1}$ and the
			algorithm moves $C_{i+1}$ to the server of $\C_i$,
		\item \label{def:doubling-decomposition:double}
			$|C_1| < 2^{t}$ and $|\C_2| = |C_1 \cup C_2| \geq 2^{t}$.
	\end{enumerate}
\end{definition}

Note that when considering a doubling decomposition, there will be exactly one
majority-vote for the components $\C_j$ --- the one after $C_1$ and $C_2$ are
merged. Thus, $C$ and all $\C_j$, $j \geq 2$, will be assigned to the server
that was picked in the majority vote of $C_1 \cup C_2$.

The following lemma shows that doubling decompositions are indeed well-defined.
Its proof provides the construction of a doubling decomposition for a given
connected component.

\begin{lemma}
\label{lem:doubling-decomposition-exists}
	Let $C$ be a connected component. Then there exists a doubling
	decomposition $(C_1, \dots, C_k)$ for $C$.
\end{lemma}
\begin{proof}
	Suppose $(u,v)$ was the last edge which caused the algorithm to set
	$C = C_u \cup C_v$. W.l.o.g.\ assume that $|C_u| \leq |C_v|$ (in case of
	ties let $C_u$ be the connected component that is moved by the algorithm). Then
	set $C_k = C_u$ and set $\C_{k-1} = C_v$. Now repeat this procedure for
	$\C_{k-1}$ in place of $C$ to obtain $C_{k-1}$ and $\C_{k-2}$. Continue
	this procedure until $C_1$ is of appropriate size.

	Note that Properties~\ref{def:doubling-decomposition:union}
	and~\ref{def:doubling-decomposition:merges} follow immediately from the
	above construction. Property~\ref{def:doubling-decomposition:move} follows
	from the definition of small-to-large steps and the choice of $C_u$ above.
	Property~\ref{def:doubling-decomposition:double} is guaranteed by the
	stopping criterion of the above recursion.
\end{proof}

Lemma~\ref{lem:majority-voting:minority-color-wrong-server} will be crucial for
the proofs of many upcoming claims in this section. The lemma asserts that when
a connected component $C$ is currently assigned to the (say) right server but at
the end it will be assigned to the left server, then it must contain relatively
many vertices that were initially assigned to the right server.

\begin{lemma}
\label{lem:majority-voting:minority-color-wrong-server}
	Let $C$ be a connected component with $|C| \geq 4$. Suppose that $C$ is
	currently assigned to server $S_i$ and that $C$ will be assigned to server
	$S_{1-i}$ when the algorithm terminates.  Then $C$ contains at least $|C|/4$
	vertices which were initially assigned to $S_i$.
\end{lemma}
\begin{proof}
	Assume w.l.o.g.\ that $C$ is currently assigned to the right server and it
	will be assigned to the left server when the algorithm terminates. We show
	that at least a $1/4$-fraction of the vertices in $C$ must be black.  This
	implies the lemma.
	
	Let $(C_1, \dots, C_k)$ be a doubling decomposition of $C$ which exists by
	Lemma~\ref{lem:doubling-decomposition-exists}.  Observe that
	$C$ must be assigned to the same server as $C_1 \cup C_2$ after they were
	merged and after the algorithm processed the majority vote for $C_1 \cup
	C_2$ (by Properties~\ref{def:doubling-decomposition:move}
	and~\ref{def:doubling-decomposition:double} of doubling decompositions).
	Thus, $C_1 \cup C_2$ had a black majority, i.e., it contains at least
	$|C_1 \cup C_2|/2$ black vertices. Since $|C_1 \cup C_2| \geq |C| / 2$, $C$
	must contain at least $|C|/4$ black vertices.
\end{proof}

Now we bound the augmentation that is used by
Algorithm~\ref{algo:majority-voting}.

\begin{lemma}
\label{lem:majority-voting:augmentation}
	The load of both servers is bounded by $n/2 + 4\Delta$. Hence,
	Algorithm~\ref{algo:majority-voting} uses at most $4\Delta$ augmentation.
\end{lemma}
\begin{proof}
	Assume that at some point during the execution of the algorithm the (w.l.o.g.)
	right server contains more vertices than the left server. We bound the
	load of the right server.

	Recall from Lemma~\ref{lem:cost-opt-2server} that
	$\Delta \leq n/4$.
	We start by considering the case where $\Delta = n/4$. 
	In this case, even moving all $n$ vertices
	to the right server only causes augmentation $n/2 = 2\Delta$.

	Now consider the case where $\Delta < n/4$.  Since $\Delta < n/4$,
	the initial assignment of $S_1$ must contain more vertices from either $V_0$
	or $V_1$. Thus, exactly one of the ground truth components $V_0$ and $V_1$
	must have a black majority (as the algorithm colored all vertices initially
	assigned to $S_1$ black).  We assume w.l.o.g.\ that $V_1$ has this black
	majority. This implies that $V_1$ has $n/2 - \Delta > n/4 > \Delta$ black
	vertices and $V_0$ has $\Delta$ black vertices. Further, as the algorithm
	proceeds, the vertices from $V_1$ must be moved to the right server.

	The right server contains at each point a (potentially empty) set of
	vertices from $V_0$ and a (potentially empty) set of vertices from $V_1$.
	For the latter set we use the trivial upper bound of $n/2$, while for the
	earlier set we give a bound of $\Delta/4$. The lemma follows.
   
	Consider a component $C$ which is on the right server and a subset of $V_0$.
	By Lemma~\ref{lem:majority-voting:minority-color-wrong-server}, $C$ contains
	at least $|C|/4$ black vertices.

	As there are only $\Delta$ black vertices in the ground truth component
	$V_0$ and each component $C \subseteq V_0$ on the right server has at least
	a $1/4$-fraction of black vertices, it follows that all components on the
	right server which are subsets of $V_0$ can only contain $4 \Delta$
	vertices.
\end{proof}

Having derived the bound for the augmentation, our next goal is to show that the
cost paid by the algorithm is bounded by $O(\alpha \Delta \lg \Delta)$. We start
by bounding the cost paid by the algorithm for each connected component.

The following lemma implies that the algorithm pays nothing for components in
which all vertices have the same color.
\begin{lemma}
\label{lem:majority-voting:same-color}
	Let $C$ be a connected component and suppose all vertices in $C$ have the
	same color. Then the algorithm has never moved the vertices in $C$.
\end{lemma}
\begin{proof}
	We prove the claim by induction over $s = |C|$.
	
	Let $|C| = s = 1$. Then $C$ consists of a single vertex. But the algorithm never
	moves single vertices unless they become part of a larger connected
	component.  Hence, $C$ is not moved.

	Now let $|C| = s+1$. Consider the last edge $(u,v)$ which was inserted that
	forced the algorithm to merge $C = C_u \cup C_v$. Since in $C$ all vertices
	have the same color, all vertices in $C_u$ and $C_v$ must have the same
	color. By induction hypothesis, the vertices in $C_u$ and $C_v$ have never
	been moved before. Thus, $C_u$ and $C_v$ must be assigned to the same
	server. This implies that a small-to-large step would not move $C_u$ or
	$C_v$. Further, a majority voting step would not move $C_u \cup C_v$ since
	all vertices vote for the server which they are already assigned to. Thus,
	no vertices in $C$ are moved.
\end{proof}

Next, we bound the cost paid for any connected component.
\begin{lemma}
\label{lem:majority-voting:server-cost}
	Let $C$ be a connected component.  Then the cost (over the entire execution
	time of the algorithm) paid for the vertices in $C$ is at most
	$O(\alpha |C| \lg |C|)$.
\end{lemma}
\begin{proof}
	Consider a vertex $u \in C$. We
	perform the following accounting: 
	we assign a token to $u$ each time when it is
	reassigned to a server and we show that the number of tokens for $u$ is
	bounded by $O(\lg |C|)$. This implies that the total number of reassignments
	for the vertices in $C$ is $O(|C| \lg |C|)$ and the lemma follows.

	First, consider the case where $u$ is moved because it is in a smaller
	connected component (Line~\ref{line:small-to-large}). Whenever this happens
	the size of the connected component containing $u$ at least doubled.  This
	can only happen $O(\lg |C|)$ times.

	Second, consider the case when $u$ is moved because of a majority vote.  A
	majority vote is performed every time when the size of the component containing
	$u$ doubled. This can only happen $O(\lg |C|)$ times and, hence, this can
	only add another $O(\lg |C|)$ tokens for $u$.

	Thus, the total number of tokens assigned to $u$ is $O(\lg |C|)$.
\end{proof}

Note that Lemma~\ref{lem:majority-voting:server-cost} is only useful for
components of size at most $O(\Delta$): If we were to apply the lemma to a
component $C$ of size $\Theta(n)$ then the cost would only be bounded by
$O(\alpha n \lg n)$.  However, this can be much worse than our desired bound of
$O(\alpha \Delta \lg \Delta)$ when $\Delta \ll n$. Thus, we need a more
fine-grained argument to obtain our goal of showing that the cost paid by
Algorithm~\ref{algo:majority-voting} never exceeds $O(\alpha \Delta \lg \Delta)$.
To do this, we first prove two technical lemmas.

\begin{lemma}
\label{lem:majority-voting:minority-color}
	Suppose $C$ is a component which is moved from $S_i$ to $S_{1-i}$ and the
	vertices in $C$ are never reassigned after this move.  Then $C$ contains
	at least $|C|/8$ vertices which were initially assigned to $S_i$.
\end{lemma}
\begin{proof}
	There are only two possible reasons why $C$ is moved: Either due to a
	small-to-large step (Line~\ref{line:small-to-large}) or due to a majority
	voting step (Line~\ref{line:majority-vote}). We consider both cases
	separately.

	Case 1: $C$ is moved due to a small-to-large step. Then by
	Lemma~\ref{lem:majority-voting:minority-color-wrong-server}, $C$ must
	contain at least $|C|/4$ vertices which were initially assigned to $S_i$.

	Case 2: $C$ is moved due to a majority voting step.

	First, consider the case when $C$ contains at most 7 vertices. Then at least
	one vertex was initially assigned to $S_i$ (if all vertices had been
	initially assigned to $S_{1-i}$, they would all have the same color and a
	majority vote would not move $C$ due to
	Lemma~\ref{lem:majority-voting:same-color}).  Thus, at least a
	$1/7$-fraction of the vertices were initially assigned to $S_i$ and the
	lemma holds.
	
	Second, suppose that $C$ contains at least $8$ vertices.  Consider the last
	edge $(u,v)$ that caused the merge $C = C_u \cup C_v$.  Suppose that the
	small-to-large step moved $C_u$ to the server of $C_v$.  Note that $C_v$
	was assigned to $S_i$ and $C_u$ was moved to $S_i$.  Now apply
	Lemma~\ref{lem:majority-voting:minority-color-wrong-server} to $C_v$. This
	implies that $C_v$ must contain at least $|C_v| / 4$ vertices that were
	initially assigned to $S_i$.  As $|C_u| \leq |C_v|$, $C$ must contain at
	least $|C|/8$ vertices that were initially assigned to $S_i$.
\end{proof}

We are now ready to show that the cost incurred by the majority-voting algorithm
never exceeds $O(\alpha \Delta \lg \Delta)$.

\begin{lemma}
\label{lem:majority-voting:cost}
	The total cost paid by Algorithm~\ref{algo:majority-voting} is at most
	$O(\alpha \Delta \lg \Delta)$ and the final assignment is a perfect
	partitioning.
\end{lemma}
\begin{proof}
	When the algorithm finishes, the final assignment must be a perfect
	partitioning because the connected components were completely revealed.  We
	only need to prove that the cost of the algorithm is
	$O(\alpha \Delta \lg \Delta)$.

	Recall that \OFF moves exactly $2\Delta$ vertices
	(Lemma~\ref{lem:cost-opt-2server}). We can assume w.l.o.g.\ that \OFF moves
	$\Delta$ vertices from $V_0$ that were initially assigned to $S_1$ to $S_0$
	and $\Delta$ vertices from $V_1$ that were initially assigned to $S_0$ to
	$S_1$.  We will argue that the cost paid by the algorithm for moving all
	vertices from $V_0$ into the $S_0$ will be $O(\alpha \Delta \lg \Delta)$;
	the same will hold for $V_1$ and $S_1$ symmetrically.
	
	Consider time $T$ during the execution of the algorithm where the following
	happens. A connected component $C$ is reassigned the left server and $C$ has
	the following properties: (1)~$C$ is a subset of $V_0$ and (2)~the vertices
	in $C$ never leave the left server after time $T$.  Since each vertex of
	$V_0$ is assigned to the left server when the algorithm terminates, each
	vertex of $V_0$ is contained in a component with the above properties (when
	a vertex or component is never moved, we set $T=0$). We call a component
	with the above properties \emph{mixed} if it contains at least one black
	vertex.  Note that when mixed component $C$ is assigned to the left server,
	$C$ contains a black vertex and, hence, $C$ must be moved from the right to
	the left server.

	We now bound the cost for mixed components. Let $X$ be the set of all mixed
	components and let $C \in X$.  	Since $C$ is mixed,
	Lemma~\ref{lem:majority-voting:minority-color} implies that at least $|C|/8$
	vertices of $C$ are black. As the black vertices in mixed components form a
	partition of the $\Delta$ black vertices in $V_0$ moved by \OFF, we obtain
	that the number of black vertices in mixed components is $\Delta$.  Thus,
	the total number of vertices in all mixed components is $\sum_{C \in X} |C|
	\leq 8 \Delta$.
	
	By Lemma~\ref{lem:majority-voting:server-cost}, the total cost paid for
	each $C \in X$ until (including) its final move is $O(\alpha |C| \lg |C|)$.
	Since (by assumption) the vertices in $C$ never move between the servers
	again, their cost never exceeds $O(\alpha |C| \lg |C|)$ until the algorithm
	finishes.  Hence, the cost paid by the algorithm for all mixed components is
	\begin{align*}
		\sum_{C \in X} O(\alpha |C| \lg |C|)
		\leq \sum_{C \in X} O(\alpha |C| \lg \Delta)
		= O(\alpha \Delta \lg \Delta).
	\end{align*}

	Now consider the vertices of $V_0$ which are not part of mixed components.
	These vertices must have been part of components in which all vertices are
	colored yellow. By Lemma~\ref{lem:majority-voting:same-color}, these
	vertices have never been moved. Thus, they do not incur any additional cost
	for the algorithm.
\end{proof}

\begin{proof}[Proof of Proposition~\ref{prop:majority-voting}]
	Lemma~\ref{lem:majority-voting:augmentation} gives the bound for the
	augmentation used by the algorithm. By Lemma~\ref{lem:majority-voting:cost}
	and Lemma~\ref{lem:cost-opt-2server}, Algorithm~\ref{algo:majority-voting}
	obtains a competitive ratio of
	\begin{align*}
	\frac{\ALG}{\OFF}
	= \frac{ O(\alpha \Delta \lg \Delta) }{ 2 \alpha \Delta }
	= O(\lg \Delta)
	= O(\lg n).
	& \qedhere
	\end{align*}
\end{proof}

\subsection{Bringing It All Together: Theorem~\ref{thm:2servers}}
\label{sec:proof-thm-2servers}

\begin{proof}{Proof of Theorem~\ref{thm:2servers}.}
Consider the following algorithm: Run the
majority-voting algorithm until we have seen all edges or until at some point it
tries to exceed the allowed augmentation. In the latter case, compute a
perfectly balanced assignment respecting the connected components and start
running Algorithm~\ref{algo:small-large-rebalance}
(Section~\ref{sec:efficient-small-large-rebalance}).

To prove the theorem, we distinguish two cases based on $\Delta$.

First, suppose $\Delta < \varepsilon n / 4$.  By
Proposition~\ref{prop:majority-voting}, Algorithm~\ref{algo:majority-voting} uses at
most $4 \Delta$ augmentation. Thus, in the current case the augmentation used by
Algorithm~\ref{algo:majority-voting} is bounded by $4 \Delta < \varepsilon n$
and it is $O(\lg n)$-competitive.  This proves the theorem for this case.

Second, suppose $\Delta \geq \varepsilon n / 4$. In this case we run
Algorithm~\ref{algo:majority-voting} until it tries to
exceed the allowed augmentation; this serves as a certificate that
$\Delta \geq \varepsilon n/4$. At this point we switch to
Algorithm~\ref{algo:small-large-rebalance}.

When we switch algorithms, Algorithm~\ref{algo:majority-voting} has
paid $O(\alpha n \lg n)$, by applying
Lemma~\ref{lem:majority-voting:server-cost}
to each connected component, and then summing over these costs. For switching to
the perfectly balanced reassignment, we only need to pay $O(\alpha n)$ \emph{once}.

By Proposition~\ref{prop:small-large-rebalance},
Algorithm~\ref{algo:small-large-rebalance} never uses more than
$O(n \lg n + (\Delta \lg n) / \varepsilon)$ vertex
moves.  Using the bound $\Delta \geq \varepsilon n /4$ and the fact
that $\OFF$ pays $2 \alpha \Delta$ (Lemma~\ref{lem:cost-opt-2server}), we
obtain the desired competitive~ratio:
\begin{align*}
	\frac{ \ALG }{ \OFF }
	&= O\left( \frac{ \alpha n \lg n + (\alpha \Delta \lg n) / \varepsilon }{\alpha \Delta}\right)
	= O\left( \frac{\lg n}{\varepsilon} \right).
	&\qedhere
\end{align*}
\end{proof}

\section{Generalization to Many Servers}
\label{sec:l-servers}

We extend our study to 
the scenario with $\ell$ servers. 
As we will see, while several concepts
introduced for the two server case 
are still useful, the $\ell$-server 
case introduces additional challenges.
We derive 
the following main result.

\begin{theorem}
\label{thm:lservers}
	Given a system with $\ell$ servers each of capacity $(1+\varepsilon)n/\ell$
	(i.e., augmentation $\varepsilon n/\ell$),
	for $\varepsilon \in (0,1/2)$,
	then there exists an $O((\ell \lg n \lg \ell) / \varepsilon)$-competitive algorithm.
\end{theorem}

Our algorithm will be  
based on a recursive bipartitioning scheme, 
described in Section~\ref{sec:bipartitioning}. 
We will use this bipartitioning
scheme to derive a static approximation algorithm 
of the optimal solution
(Section~\ref{sec:offline}). Then we provide 
a recursive version of the majority
voting algorithm which we will compare against the approximation algorithm
(Section~\ref{sec:recursive-majority-voting}). In
Section~\ref{sec:small-large-rebalance-lservers}, we analyze the
Small--Large--Rebalance algorithm in the $\ell$ server setting and
we conclude 
by proving Theorem~\ref{thm:lservers} in
Section~\ref{sec:proof-thm-lservers}

\subsection{The Bipartition Tree}
\label{sec:bipartitioning}

We establish a recursive bipartitioning scheme of the
servers which we will be using throughout the rest of this section. All
algorithms in this section which use the recursive bipartitioning create such a
bipartitioning at the start of the algorithm, before the adversary provides any
edge. After that the bipartitioning will never be changed.

We obtain the bipartition scheme by growing a balanced binary tree on a set of
$\ell$ leaves, where each leaf corresponds to a server $S_i$. We call this tree
the \emph{bipartition tree} and denote it by $\T$.

We denote the internal nodes of $\T$ by $w_1, \dots, w_s$. For an internal node
$w_j$, we write $T(w_j)$ to denote the subtree of $T$ which is rooted at $w_j$
and we define $S(w_j)$ to be the set of servers which are leaves in
$T(w_j)$.  We further write $V(w_i)$ to denote the set of vertices which are
assigned to the servers in $S(w_j)$. See Figure~\ref{fig:bipartition-tree} for an
illustration.

\begin{figure}
  \begin{centering}
    \includegraphics[width=.99\columnwidth]{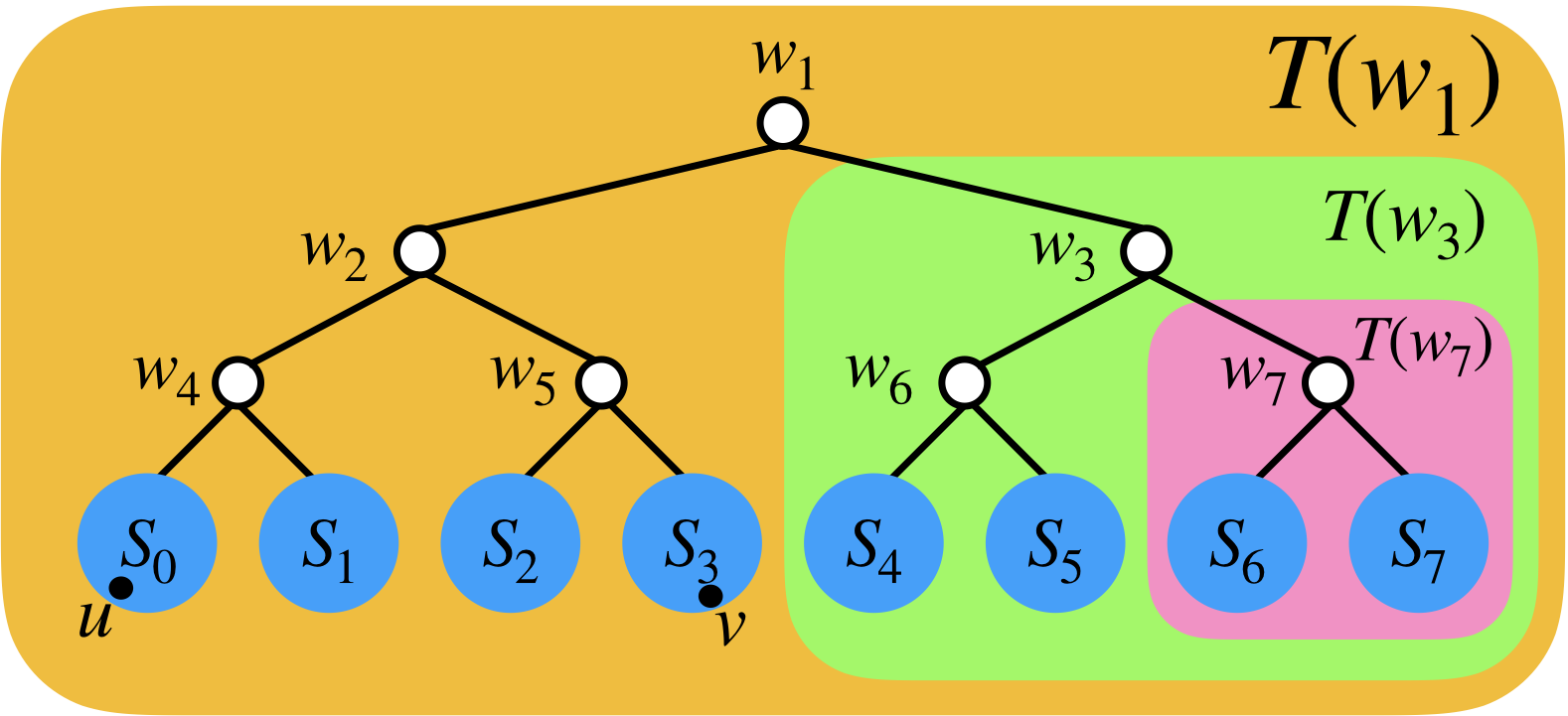}
	\caption{ An illustration of the bipartition tree $\T$ for servers
			$S_0,\dots,S_7$. The internal nodes of the bipartition tree are
			denoted $w_1,\dots,w_7$. We highlighted the subtrees $\T = T(w_1)$,
			$T(w_3)$, and $T(w_7)$. Here we obtain the server sets $S(w_1) = \{S_0,\dots,S_7\}$,
			$S(w_3) = \{S_4,\dots,S_7\}$, and $S(w_7) = \{S_6,S_7\}$.  }
	\label{fig:bipartition-tree}
  \end{centering}
\end{figure}

Observe that $\T$ defines a bipartition scheme: let $w$ be an internal node of
$\T$ and let $w_0$, $w_1$ be its children. Then\footnote{If $w_j$ is a
	leaf corresponding to server $S$, we set $S(w_j) = \{S\}$.}
$S(w_0)$ and $S(w_1)$ are disjoint and their union is $S(w)$. Thus, $\T$ implies
a bipartition scheme of the servers and internal nodes correspond to
bipartition~steps.

Note that since $\T$ is a balanced binary tree, there are $\ell-1$ internal nodes
in total and each server is contained in at most $\lceil \lg \ell \rceil$ subtrees
of $T$. Hence, for each server $S_j$ there are at most $\lceil \lg \ell \rceil$
internal vertices $w$ such that $S_j \in S(w)$.

In the following we will refer to the internal nodes in $\T$ as \emph{nodes},
whereas the vertices $V$ from the graph $G$ are called \emph{vertices}.

\subsection{Offline Approximation Algorithm}
\label{sec:offline}

We are not aware of a concise characterization of the optimal solution used by
\OFF (unlike in the two-server case in Section~\ref{sec:cost-opt}).  Thus, to
get a better understanding of the solution obtained by \OFF, we provide an
\emph{offline} approximation algorithm, called \APPROXALGO, which exploits the
previously defined bipartition scheme and which obtains a 2-approximation of the
optimal solution.  However, unlike the solution obtained by \OFF, we allow the
approximation algorithm to use \emph{unlimited} augmentation in each server; its
only goal is to move all vertices from the same ground truth components to the
same server using few vertex moves.\footnote{In this setting, a trivial solution
	assigns all vertices to the same server at cost $O(\alpha n)$.}
Later, \APPROXALGO will play a role for the design and analysis of our
online algorithm.

Intuitively, \APPROXALGO traverses the bipartition tree $\T$ top--down and
greedily minimizes the number of vertices ``moved over'' each server
bipartition.  We now describe the algorithm in more detail.

\APPROXALGO is given the sequence of edges $\sigma = (e_1, \dots, e_r)$
\emph{a priori} and it also knows the initial assignment $\init(S_0), \dots,
\init(S_{\ell-1})$ of the vertices to the $\ell$ servers.  Using the knowledge
about the edges, \APPROXALGO starts by computing the connected components of $G$
and obtaining the ground truth components $V_0, \dots, V_{\ell-1}$. 

Now, for each ground truth component $V_i$, \APPROXALGO does the following.  Let
$\root$ be the root of $\T$ and let $w_0$ and $w_1$ denote its children.  Let
$A_{ij} = V(w_j) \cap V_i$, $j = 0,1$, denote the vertices from $V_i$ which are
currently assigned to servers in $S(w_j)$. Define $n_{ij} = |A_{ij}|$ and assume
w.l.o.g.\ that $n_{i0} \geq n_{i1}$. The algorithm marks the vertices from
$A_{i1}$ as \emph{dirty}. Now the algorithm recurses on the subtree $T(w_0)$ in
place of $\T$ and marks more vertices of $V_i$ as dirty.  The recursion stops
when $S(w_0)$ only contains a single server $S$.  Then the algorithm moves all
dirty vertices of $V_i$ into server $S$. 

The pseudocode of \APPROXALGO is stated in Algorithm~\ref{algo:offline}.

\begin{algorithm}[t!]
\caption{The static approximation algorithm \APPROXALGO}\label{algo:offline}
\begin{algorithmic}[1]
\Input{All edges $e_1, \dots, e_r$ at once}
\Procedure{\APPROXALGO}{$e_1, \dots, e_r$}
	\State Compute the connected components $V_0, \dots, V_{\ell-1}$ of $G$
	\For{$i = 0, \dots, \ell-1$}
		\State \textsc{RecursiveStep}($\T$, $i$)
	\EndFor
\EndProcedure

\Procedure{RecursiveStep}{$T$, $i$}
	\If{$T$ contains only a single server $S$}
		\State Move all dirty vertices of $V_i$ into $S$
		\State \Return
	\EndIf
	\State Let $\root$ be the root of $T$ and denote its children $w_0$, $w_1$
	\State $A_{ij} \leftarrow V(w_j) \cap V_i$, $j = 0,1$
	\State $n_{ij} \leftarrow |A_{ij}|$ and suppose w.l.o.g.\ that $n_{i0} \geq n_{i1}$
	\State Mark all vertices in $A_{i1}$ dirty
	\State \textsc{RecursiveStep}($T(w_0)$, $i$)
\EndProcedure
\end{algorithmic}
\end{algorithm}

By overloading notation, we let \APPROXALGO denote the cost paid by \APPROXALGO.
Further, we let $\APPROXALGO_i$ denote the cost paid by \APPROXALGO to move all
vertices from $V_i$ to the same server $S$.

We now show that \APPROXALGO indeed yields a 2-approximate solution of the cost
of the optimal offline algorithm.

\begin{lemma}
\label{lem:offline-cost}
	$\APPROXALGO \leq 2 \cdot \OFF$.
\end{lemma}
\begin{proof}
	Fix any $i \in \{0,\dots,\ell-1\}$. Let $\OFF_i$ denote the cost paid
	by \OFF to move the vertices from $V_i$ to the same server. We show that
	$\APPROXALGO_i \leq 2 \cdot \OFF_i$. This claim implies the lemma since
	\begin{align*}
		\APPROXALGO = \sum_i \APPROXALGO_i \leq 2 \sum_i \OFF_i = 2 \cdot \OFF.
	\end{align*}

	Observe that while $\APPROXALGO_i$ proceeds, it traverses $\T$ from root
	$\root$ to one of the leaves, and at each step, it increases the level of
	the current internal node by one. 

	Using the solution of $\OFF_i$, we can define a similar traversal of $\T$:
	Let $\root$ be the root of $\T$ and let $w_0$ and $w_1$ be its children.  As
	$\OFF_i$ must move all vertices from $V_i$ to the same server $S$, $\OFF_i$
	moves the vertices from $A_{ij} = V(w_j) \cap V_i$ to a server $S$ in
	$S(w_{1-j})$ for $j\in\{0,1\}$.  We call the moved vertices \emph{dirty}.
	After this move,  $\OFF_i$ still needs to process the vertices of $V_i$
	which were initially assigned to a server in $S(w_{1-j})$ but not to $S$.
	We can view this as processing $T(w_{j-1})$. Thus, $\OFF_i$ traverses $\T$
	until the final server $S$ is reached and marks a subset of $V_i$ dirty.

	The previous paragraphs define to two different traversals of $\T$ and two
	different sets of dirty vertices. Let $h$ be the smallest level where the
	two traversals picked different internal nodes in $\T$.

	Until level $h-1$, both vertices have marked the same $W$ vertices dirty.
	At levels $h$ and below, we obtain the following bounds. Let $w$ be the
	internal node at level $h-1$ that is traversed by both algorithms and let
	$w_0$, $w_1$ denote its children at level $h$.  Let $n_{ij} = |V(w_j) \cap
	V_i|$ be defined as in the definition of \APPROXALGO.  $\APPROXALGO_i$ marks
	at most $|V(w)| = n_{i0} + n_{i1}$ vertices from $V(w)$ as dirty.  Since the
	two traversals split at level $h$ and $\APPROXALGO_i$ moves $n_{i1} \leq
	n_{i0}$ vertices (by definition), $\OFF_i$ moves at least $n_{i0}$ vertices.

	Recall that for each algorithm, its sets of dirty vertices and its set of
	moved vertices are identical. Now the following computation proves the claim
	that $\APPROXALGO_i \leq 2 \cdot \OFF_i$:
	\begin{align*}
		\frac{\APPROXALGO_i}{\OFF_i}
		\leq \frac{\alpha(W + n_{i0} + n_{i1})}{\alpha(W + n_{i0})}
		\leq \frac{W + 2n_{i0}}{W + n_{i0}}
		\leq 2.
		& \qedhere
	\end{align*}
\end{proof}

\subsection{The Recursive Majority Voting Algorithm}
\label{sec:recursive-majority-voting}

We now describe an algorithm which works efficiently in the setting with $\ell$
servers whenever \OFF does not perform too many vertex moves.  The algorithm can
be viewed as a generalization of Algorithm~\ref{algo:majority-voting} to $\ell$
servers, by exploiting the previously defined bipartitioning scheme.

\subsubsection{The Algorithm}
The algorithm consists of two parts: A single \emph{global algorithm} and
multiple \emph{local algorithms}, one per internal node in $\T$. The global
algorithm maintains a recursive bipartitioning scheme (as defined in
Section~\ref{sec:bipartitioning}) and runs a local algorithm on each
bipartition. The local algorithms are used to ``reduce'' the setting with
multiple servers to the case with two servers.

We now describe the two parts in more detail and state the pseudocode in
Algorithm~\ref{algo:recursive-majority-voting}. We write $S_u$ to denote the
server which vertex $u$ is assigned to.

\paragraph{Global Algorithm.}
The global algorithm starts by computing the bipartition tree $\T$. On each
internal node $w$ of $\T$, the global algorithm instantiates a local algorithm
which we describe below.

Furthermore, the global algorithm iterates over all vertices and does the following
for each $v \in V$. The algorithm finds all internal nodes $w_i$ such that $v
\in V(w_i)$ and labels $v$ with $w_i$.  This labelling of the vertices only
takes into account the initial assignment of the vertices and will never be
changed throughout the running time of the algorithm.  For example, if the
vertices $u$ and $v$ in Figure~\ref{fig:bipartition-tree} are assigned to
servers $S_0$ and $S_3$ in the initial assignment, their labels will be
$\{w_1,w_2,w_4\}$ and $\{w_1,w_2,w_5\}$, respectively.

When the adversary provides an edge $(u,v)$, the global algorithm does the
following.  It locates the servers $S_u$ and $S_v$. If $S_u = S_v$, the
algorithm merges the components and continues with the next edge. If
$S_u \neq S_v$, the global algorithm finds the internal node $w$ in $\T$ which
is the lowest common ancestor of $S_u$ and $S_v$.  (For example, in
Figure~\ref{fig:bipartition-tree} the lowest common ancestor for $u$ and $v$ is
$w_2$.) Then the global algorithm gives the edge $(u,v)$ to the local
algorithm corresponding to $w$.

\paragraph{Local Algorithms.}
A local algorithm is run on an internal node $w$ of $\T$.  Let $w_0$ and $w_1$
denote the children of $w$ in $\T$.  Note that each local algorithm corresponds
to a bipartition step where the servers in $S(w)$ are partitioned into subsets
$S(w_0)$ and $S(w_1)$.

An instance of the local algorithm only receives edges $(u,v)$ from the global
algorithm when (1)~their endpoints are assigned to servers $S_u, S_v \in S(w)$
and (2)~$S_u$ and $S_v$ are in different sets of the bipartition, i.e., $S_u \in
S(w_j)$ and $S_v \in S(w_{1-j})$.

When the global algorithm provides an edge $(u,v)$ with the above properties,
the local algorithm locates $C_u, C_v$, $S_u$ and $S_v$. Assume
w.l.o.g.\ that $|C_u| \leq |C_v|$.  Then $C_u$ is moved to $S_v$ and
$C_u$ and $C_v$ are merged.\footnote{When $C_u$ changes its server, all local
algorithms corresponding to internal nodes $w$ with $S_u \in S(w)$ or
$S_v \in S(w)$, must be informed about this move. This can be done by
recomputing $V(w)$ for each internal node $w$.  Note that
this is just an internal operation of the data structure and does not incur
any cost to the algorithm.} As before, we call
this a \emph{small-to-large step}.

Finally, the local algorithm checks whether the new component $C_u \cup C_v$ has
size $n/\ell$ or it surpassed a power of $2$, i.e., it checks if $|C_u \cup C_v| =
n/\ell$ or there exists an $i \in \mathbb{N}$ s.t.\ $|C_u| < 2^i$, $|C_v| < 2^i$
and $|C_u \cup C_v| \geq 2^i$. If this is the case, the local algorithm triggers
a majority voting step for $C_u \cup C_v$ which we explain next.

\paragraph{Majority Voting Step.}
When a local algorithm triggers a majority voting step for a connected component
$C$, the algorithm does the following. Let $\root$ be the root of $\T$ and let $w_0$
and $w_1$ be the two child nodes of $\root$. For $j \in \{0,1\}$, let $n_j$ denote the number of
vertices in $C$ with label $w_j$. If $n_j \geq n_{1-j}$,
the algorithm recurses on $w_j$ in place of $w$; else, the algorithm recurses on
$w_{1-j}$ in place of $w$. The recursion continues until a leaf in the
bipartitioning tree is reached which corresponds to a server $S$. Then the
algorithm moves $C$ to $S$.

Note that the above majority voting procedure is very similar to what
\APPROXALGO does for a single ground truth component $V_i$.

\begin{algorithm}[t!]
\caption{The Recursive Majority Voting Algorithm}\label{algo:recursive-majority-voting}
\begin{algorithmic}[1]
\Input{A sequence of edges $\sigma=(e_1, \dots, e_r)$}
\Procedure{GlobalAlgorithm}{$e_1, \dots, e_r$}
	\State Create a bipartition tree $\T$			 \Comment{Initialization phase}
	\For{each internal node $w$ of $\T$}
		\State Instantiate \textsc{LocalAlgorithm}($w$)
	\EndFor
	\For{$v \in V$}
		\State Label $v$ with each internal node $w$ of $\T$ s.t.\ $v \in V(w)$
	\EndFor
	\For{$i = 1, \dots, r$}	\Comment{Processing of the edges}
		\State $(u,v) \leftarrow e_i$
		\If{$S_u = S_v$}
			\State Merge $C_u$ and $C_v$, \textbf{continue}
		\EndIf
		\State $w \leftarrow$ the lowest common ancestor of $S_u$ and $S_v$ in $\T$
		\State \textsc{LocalAlgorithm}($w$, $(u,v)$)
	\EndFor
\EndProcedure

\Procedure{LocalAlgorithm}{$w$, $(u,v)$}
	\State $w_0, w_1 \leftarrow$ the children of $w$ in $\T$
	\State Suppose w.l.o.g.\ that $|C_u| \leq |C_v|$
	\State Check if moving $C_u$ to $S_v$ triggers the stopping criterion
	\State Move $C_u$ to $S_v$ and merge $C_u$ and $C_v$
			\Comment{Small-to-large step}
	\If{$|C_u \cup C_v| = n/\ell$ \textbf{or}
			there exists an $i \in \mathbb{N}$
			s.t.\ $|C_u| < 2^i$, $|C_v| < 2^i$ and $|C_u \cup C_v| \geq 2^i$}
		\State \textsc{MajorityVotingStep}($C_u \cup C_v$)
	\EndIf
\EndProcedure

\Procedure{MajorityVotingStep}{$C$}
	\State $\root \leftarrow$ the root of $\T$
	\State\label{line:majority-voting-recursion} $w_0, w_1 \leftarrow$ the
		children of $\root$ in $\T$
	\State $n_j \leftarrow$ the number of vertices labeled with $w_j$ in $C$, $j=0,1$
	\If{$n_0 \geq n_1$}
		$\root \leftarrow w_0$
		\textbf{else} 
		$\root \leftarrow w_1$
	\EndIf
	\If{$S(r)$ contains only one server}
		\State Check if moving $C$ to $S(r)$ triggers the stopping criterion
		\State Assign $C$ the single server in $S(r)$
	\Else \, Go to Line~\ref{line:majority-voting-recursion}
	\EndIf	
\EndProcedure
\end{algorithmic}
\end{algorithm}

\paragraph{Stopping Criterion.}
To ensure that the algorithm does not exceed the augmentation of the
servers, we add a stopping criterion.

To define the stopping criterion, let $w$ be an internal node of $\T$ with
children $w_0$ and $w_1$. For $j \in \{0,1\}$, we call $w_j$
\emph{overloaded} if $V(w_j)$ contains at least
$\varepsilon n/(\ell \lceil \lg \ell \rceil)$ vertices with label
$w_{1-j}$.

Intuitively, the condition states that an internal
node $w_j$ is overloaded when its servers $S(w_j)$ obtained ``many'' vertices
which were initially assigned to the other side of the bipartition, $S(w_{1-j})$.

The \emph{stopping criterion} is checked before each component move (i.e.,
before each small-to-large step and before each majority voting step). It is
triggered if the component move would create an assignment in which there exists
an overloaded internal node $w$.  When the stopping criterion is triggered, the
global algorithm and all local algorithms stop and
Algorithm~\ref{algo:small-large-rebalance} is started instead (we show in
Section~\ref{sec:small-large-rebalance-lservers} that
Algorithm~\ref{algo:small-large-rebalance} also works for $\ell$ servers).

\subsubsection{Structural Properties}
\label{sec:structural-properties}
To obtain a better understanding of the algorithm, we first prove some
structural properties about it and defer its cost analysis to
Section~\ref{sec:recursive-majority-voting:analysis}. We consider the setting
where each server has capacity $(1+\varepsilon)n/\ell$ for $\varepsilon \in
(0,1/2)$.

In Subsections~\ref{sec:structural-properties}
and~\ref{sec:recursive-majority-voting:analysis}, we only analyze the cost
Algorithm~\ref{algo:recursive-majority-voting} without the cost of
Algorithm~\ref{algo:small-large-rebalance}.  We analyze the cost of
Algorithm~\ref{algo:small-large-rebalance} for $\ell$ servers in
Sections~\ref{sec:small-large-rebalance-lservers}
and~\ref{sec:proof-thm-lservers}.

We begin by showing that as long as the stopping criterion is not triggered, the
vertex assignment created by Algorithm~\ref{algo:recursive-majority-voting} is
close to the initial assignment.
\begin{lemma}
\label{lem:recursive-majority-voting:augmentation-light}
	Suppose the stopping criterion is not triggered. Then:
	\begin{enumerate}
	\item Each server contains at most $\varepsilon n/\ell$ vertices that were
			not initially assigned to it.
	\item Each server contains at least $(1-\varepsilon) n/\ell$ vertices that
			were initially assigned to it.
	\end{enumerate}
\end{lemma}
\begin{proof}
	Consider any server $S_j$. We show that since the stopping criterion is not
	triggered, $V(S_j)$ obtains at most
	$\varepsilon n / (\ell \lceil \lg \ell \rceil)$ vertices for each of the
	$\lceil \lg \ell \rceil$ subtrees in $\T$ containing $S_j$.
	
	As argued in Section~\ref{sec:bipartitioning}, there are at most
	$\lceil \lg \ell \rceil$ internal nodes $w$ of $\T$ such that $S_j \in S(w)$.
	Since the stopping criterion is not triggered, no internal node of $\T$ is
	overloaded.
	
	To prove Part~(1), consider an internal node $w$ of $\T$ with $S_j \in
	S(w)$. Let $w_0, w_1$ be the children of $w$ and suppose $S_j \in S(w_r)$.
	Observe that $V(S_j)$ can obtain at most
	$\varepsilon n / (\ell \lceil \lg \ell \rceil)$ vertices that were originally
	assigned to servers in $S(w_{1-r})$ (if it had received more vertices, then
	$w_{1-r}$ would be overloaded). As there are at most $\lceil \lg \ell \rceil$
	nodes $w$ with the above property, the number of vertices which were not
	initially assigned to $S_j$ is bounded by $\varepsilon n / \ell$.

	Now let us prove Part~(2). Consider an internal node $w$ of
	$\T$ with $S_j \in S(w)$. Let $w_0, w_1$ be the children of $w$ and suppose
	$S_j \in S(w_r)$. Now observe that the servers in $S(w_{1-r})$ can have
	obtained $\varepsilon n / (\ell \lceil \lg \ell \rceil)$ vertices that were
	originally assigned to $S_j$ (if they had received more vertices, then
	$w_{1-r}$ would be overloaded). As there are at most $\lceil \lg \ell \rceil$
	nodes $w$ with the above property, it follows that the number of vertices
	assigned to servers $\{S_0, \dots, S_{\ell-1}\} \setminus \{ S_{j} \}$ that
	were initially assigned to $S_j$ is $\varepsilon n / \ell$. Hence, $S_j$
	must contain at least $(1-\varepsilon)n/\ell$ vertices that were initially
	assigned to it.
\end{proof}

As a corollary of Lemma~\ref{lem:recursive-majority-voting:augmentation-light}
we obtain the following lemma.
\begin{lemma}
\label{lem:recursive-majority-voting:augmentation}
	(1) As long as the stopping criterion is not triggered, the load of each server
	is bounded by $(1+\varepsilon)n/\ell$, i.e.,
	Algorithm~\ref{algo:recursive-majority-voting} uses only $\varepsilon n/\ell$
	augmentation.

	(2) When the stopping criterion is triggered, the augmentation still does
	not exceed $\varepsilon n/\ell$.
\end{lemma}
\begin{proof}
	Part~(1) of the lemma follows immediately from Part~(1) of
	Lemma~\ref{lem:recursive-majority-voting:augmentation-light}. Let us prove
	Part~(2): The stopping criterion is checked every time before a component is
	moved. Hence, at the time when the algorithm checks the stopping criterion,
	the algorithm did not exceed the augmentation bound due to Part~(1). If the
	algorithm triggers the stopping criterion, then the component was not yet
	moved and the augmentation is still the same as before.
\end{proof}

Define the \emph{final assignment} to be the assignment which is created by
Algorithm~\ref{algo:recursive-majority-voting} once it has seen all edges in
$G$.  We show that the final assignment of the algorithm provides a perfect
partitioning if the stopping criterion is not triggered.

\begin{lemma}
\label{lem:recursive-majority-voting:perfect-partitioning}
	If Algorithm~\ref{algo:recursive-majority-voting} stops and the stopping
	criterion is not triggered, then the final assignment is a perfect
	partitioning.
\end{lemma}
\begin{proof}
	By definition of the algorithm, vertices of the same connected component are
	always assigned to the same server. When the algorithm finishes, all edges
	of $G$ were revealed and each component has size $n/\ell$.  By
	Lemma~\ref{lem:recursive-majority-voting:augmentation-light}, the augmentation of
	each server is at most $\varepsilon n/\ell$. Since
	$\varepsilon < 1/2$, no server can have more than one component
	assigned. As each component is placed on a server, each
	component is placed alone on a server. This proves that the
	algorithm creates a perfect partitioning.
\end{proof}

Indeed, we show that the final assignment of
Algorithm~\ref{algo:recursive-majority-voting} is not only a perfect
partitioning, but it is the same assignment as the one created by \APPROXALGO
from Section~\ref{sec:offline}.
\begin{lemma}
\label{lem:recursive-majority-voting:same-as-offline}
	If Algorithm~\ref{algo:recursive-majority-voting} stops and the stopping
	criterion is not triggered, Algorithm~\ref{algo:recursive-majority-voting}
	and \APPROXALGO have the same final assignment.
\end{lemma}
\begin{proof}
	By Part~(2) of Lemma~\ref{lem:recursive-majority-voting:augmentation-light},
	Algorithm~\ref{algo:recursive-majority-voting} moves at most
	$\varepsilon n/\ell$ vertices out of each server compared to the initial
	assignment. Hence, in the final assignment each server must still contain
	at least $(1-\varepsilon)n/\ell > n/(2\ell)$ vertices from its original assignment
	since $\varepsilon \in (0,1/2)$. Thus, in the final assignment
	each server contains more than half of the vertices that were originally
	assigned to it.

	Consider any server $S_j$ and let $\init(S_j)$ be the set of vertices
	initially assigned to $S_j$. Then there must exist
	a ground truth component $V_i$ with $|V_i \cap \init(S_j)| \geq n/(2\ell)$.
	We show that \APPROXALGO and
	Algorithm~\ref{algo:recursive-majority-voting} both assign this component $V_i$ to $S_j$. This
	proves the lemma since this claim holds for any $S_j$.
	
	First, consider \APPROXALGO. Note that at each step of the traversal of
	$\T$, the majority of the vertices in $V_i$ will vote for the internal node
	containing server $S_j$. Hence, \APPROXALGO will place $V_i$ on $S_j$.

	Second, consider Algorithm~\ref{algo:recursive-majority-voting}. When the
	algorithm stops, all edges were revealed and the connected components agree
	with the ground truth components. Now consider the component $C = V_i$. When
	the $C$ grows to size $|C| = n/\ell$, the algorithm performs a majority
	voting step (by definition of the algorithm). At this point, more than half
	of the vertices in $C$ were labeled with $S_j$ (because more than half of
	the vertices from $C = V_i$ were originally assigned to $S_j$).  Hence,
	Algorithm~\ref{algo:recursive-majority-voting} will also place $V_i$ on
	$S_j$.
\end{proof}

\subsubsection{Analysis}
\label{sec:recursive-majority-voting:analysis}

The rest of this subsection is devoted to proving the following proposition
about Algorithm~\ref{algo:recursive-majority-voting}.

\begin{proposition}
\label{prop:recursive-majority-voting}
	Suppose there are $\ell$ servers and each has capacity
	$(1+\varepsilon)n/\ell$
	for $\varepsilon \in (0,1/2)$, i.e., the augmentation is $\varepsilon n/\ell$.
	Algorithm~\ref{algo:recursive-majority-voting} has the following properties:
	\begin{enumerate}
		\item If the stopping criterion is not triggered, the algorithm creates
		a perfect partitioning, its cost is bounded by $O(\OFF \cdot \lg n)$ and
		at no point during its execution it uses more than $\varepsilon n/\ell$
		augmentation.
		\item If the stopping criterion is triggered, the cost of the algorithm
		is $O(\alpha n \lg n)$ plus the cost of
		Algorithm~\ref{algo:small-large-rebalance} and the cost of \OFF is at
		least $\Omega(\alpha \varepsilon n / (\ell \lg \ell))$.
	\end{enumerate}
\end{proposition}

We prove the proposition at the end of this section.  We start by proving a
sequence of lemmata and begin by reasoning about the cost paid by
Algorithm~\ref{algo:recursive-majority-voting}. As shown in
Lemma~\ref{lem:reduction} we only need to bound the moving cost paid by
Algorithm~\ref{algo:recursive-majority-voting} to bound its total cost.

The following lemma bounds the cost paid for any connected component $C$.
\begin{lemma}
\label{lem:recursive-majority-voting:server-cost}
	Let $C$ be a connected component. Then the cost (over the entire execution
	time of the algorithm) paid for moving the vertices
	in $C$ is $O(\alpha |C| \lg |C|)$. 
\end{lemma}
\begin{proof}
	We can use the same accounting argument as in the proof of
	Lemma~\ref{lem:majority-voting:server-cost}. That is, we assign a token to
	a vertex $v$ whenever it is moved. Now, whenever the component $C$
	containing $v$ is moved due to a small-to-large step, the size of $C$
	doubles. This can only happen $O(\lg |C|)$ times. Furthermore, there are only
	$O(\lg |C|)$ majority voting steps involving $u$: Each majority voting step
	is triggered because $|C| = n/\ell$ or because $|C|$ surpassed a power of $2$;
	the first event can happen only once and the second event can happen at most
	$O(\lg |C|)$ times. Hence, $v$ will never accumulate more than $O(\lg |C|)$
	tokens.  Since the above arguments apply for each $v \in C$, the total cost
	paid for moving the vertices in $C$ is bounded by $O(|C| \lg |C|)$.
\end{proof}

Let $f \colon V \to \{0,\dots,\ell-1\}$ be the function which maps each vertex
to its server in the final assignment by
Algorithm~\ref{algo:recursive-majority-voting}. That is, when
Algorithm~\ref{algo:recursive-majority-voting} processed all edges,
each $v$ is assigned to $S_{f(v)}$. For a connected component $C$, set $f(C) =
f(u)$ for $u \in V$.  Note that $f(C)$ is well-defined since all
vertices of $C$ are assigned to the same $S_{f(C)}$ when the
algorithm terminates.

In the following proofs, we will write $\hw{w}{C}$ to denote the number of
vertices in a connected component $C$ which are labeled with $w$. We further
write $\bhw{w}{C}$ to denote the number of vertices in $C$ which are \emph{not}
labeled with $w$, i.e., $\bhw{w}{C} = |C| - \hw{w}{C}$.

Lemma~\ref{lem:recursive-majority-voting:colorful-components} shows that
whenever a component $C$ is assigned to a server which is not its final server,
it must contain relatively many vertices which were not initially assigned to
its final server $S_{f(C)}$.

\begin{lemma}
\label{lem:recursive-majority-voting:colorful-components}
	Consider any point in the execution of the algorithm at which a connected
	component $C$ is assigned to server $S \neq S_{f(C)}$. Let $w$ be the lowest
	common ancestor of $S$ and $S_{f(C)}$ in $\T$ and denote the children of $w$
	by $w_0$ and $w_1$.
	
	If $S_{f(C)} \in S(w_j)$ for $j \in \{0,1\}$, then:
	\begin{enumerate}
		\item $C$ contains at least $|C|/4$ vertices which do not have label
		$w_j$, i.e., $\bhw{w_j}{C} \geq |C|/4$.
		\item $C$ contains at least $\bhw{w_j}{C}$ vertices which were not initially
		assigned to $S_{f(C)}$.
	\end{enumerate}
\end{lemma}
\begin{proof}
	To prove Part~(1), 
	consider a doubling decomposition $(C_1, \dots, C_k)$ of $C$ (see
	Definition~\ref{def:double-decomposition}); the decomposition exists by
	Lemma~\ref{lem:doubling-decomposition-exists} which also applies in the
	$\ell$ server setting.  After $C_1$ and $C_2$ were merged,
	Algorithm~\ref{algo:recursive-majority-voting}
	performed a majority voting step and placed $C_1 \cup C_2$ in a server $S
	\in S(w_{1-j})$. Thus, $\hw{w_j}{C_1 \cup C_2} \leq |C_1 \cup C_2|/2$
	(otherwise, the majority voting step would have chosen a server in
	$S(w_j)$).  Since
	$|C_1 \cup C_2| \geq |C|/2$ and $C_1 \cup C_2 \subseteq C$,
	\begin{align*}
		\bhw{w_j}{C}
		&\geq \bhw{w_j}{C_1 \cup C_2}
		= |C_1 \cup C_2| - \hw{w_j}{C_1 \cup C_2} \\
		&\geq |C_1 \cup C_2| - |C_1 \cup C_2|/2
		= |C_1 \cup C_2|/2 \geq |C|/4.
	\end{align*}

	For Part~(2) note that each vertex which was initially assigned to
	$S_{f(C)}$ has label $w_j$ (because $S_{f(C)} \in S(w_j)$ by assumption).
\end{proof}

In the following, we show that the cost paid by the algorithm is $O(\OFF \cdot \lg n)$
when the stopping criterion is not triggered. We start by showing that when a
component is moved for the last time, it contains a large number of vertices
which did not originate from the server it is assigned to.

\begin{lemma}
\label{lem:recursive-majority-voting:minority-color}
	Let $C$ be a component which is moved to server $S_{f(C)}$ and suppose the
	vertices of $C$ are never reassigned after this move.\footnote{Note that
		when a small-to-large step is performed, two components are merged due
		to the corresponding edge insertion. In this case, the component $C$ in
		the lemma is the component which is being moved (i.e., before merging).}
	Then $C$ contains at least $|C|/8$ vertices which were not assigned to
	$S_{f(C)}$ in the initial assignment.
\end{lemma}
\begin{proof}
	Note that $C$ is moved due to one of two reasons: Either
	because of a small-to-large step or because of a majority voting step. We
	distinguish between these cases.

	In case of a small-to-large step, $C$ is assigned to a server $S \neq
	S_{f(C)}$ before the move.
	Lemma~\ref{lem:recursive-majority-voting:colorful-components} implies
	that $C$ contains at least $|C|/4$ vertices which were not originally
	assigned to $S_{f(C)}$.

	Now suppose that $C$ is moved due to a majority voting step. Let $(u,v)$ be
	the last edge which was inserted and which triggered the majority voting
	step for $C$. Then Algorithm~\ref{algo:recursive-majority-voting} previously
	merged components $C_u$ and $C_v$; suppose w.l.o.g.\ that $C_u$ was moved to
	$C_v$ and $|C_u| \leq |C_v|$. Prior to the majority voting step, $C$ is
	assigned to the same server $S \neq S_{f(C)}$ that $C_v$ was assigned to
	before $(u,v)$ was inserted. Hence, we can apply
	Lemma~\ref{lem:recursive-majority-voting:colorful-components} to $C_v$ and
	obtain that $C_v$ contains at least $|C_v|/4$ vertices which were not
	initially assigned to $S_{f(C)}$. Thus, the number of vertices
	in $C$ which do not originate from $S_{f(C)}$ is at least
	\begin{align*}
		|C_v|/4 \geq 2|C_v|/8 \geq |C_u \cup C_v| / 8 = |C|/8.
		&\qedhere
	\end{align*}
\end{proof}

The next lemma considers the cost paid by
Algorithm~\ref{algo:recursive-majority-voting} when the stopping criterion is
not triggered.
\begin{lemma}
\label{lem:recursive-majority-voting:total-cost}
	Suppose there are $\ell$ servers and each has capacity
	$(1+\varepsilon)n/\ell$
	for $\varepsilon \in (0,1/2)$, i.e., the augmentation is $\varepsilon n/\ell$.
	If the stopping criterion is not triggered and
	Algorithm~\ref{algo:recursive-majority-voting} stops, then the cost paid by
	the algorithm is $O(\OFF \cdot \lg n)$.
\end{lemma}
\begin{proof}
	Fix some $i \in \{0,\dots,\ell-1\}$. Recall that $\APPROXALGO_i$ denotes the
	cost paid by \APPROXALGO to move the vertices from $V_i$ to the server
	$S_{f(V_i)}$. We show that for $V_i$,
	Algorithm~\ref{algo:recursive-majority-voting} pays $O(\APPROXALGO_i \lg n)$.
	The lemma follows from this claim and Lemma~\ref{lem:offline-cost}, since
	the total cost paid by Algorithm~\ref{algo:recursive-majority-voting} is
	bounded by
	\begin{align*}
		\sum_i O(\APPROXALGO_i \cdot \lg n) = O(\APPROXALGO \cdot  \lg n) = O(\OFF \cdot \lg n).
	\end{align*}

	Consider any ground truth component $V_i$ and let $\Delta$ denote the number
	of vertices $\APPROXALGO_i$ moves to server $S_{f(C)}$. Note that as
	$\APPROXALGO_i$ moves $\Delta$ vertices into $S_{f(C)}$, we get
	$\APPROXALGO_i = \alpha \Delta$. 

	Consider time $T$ of the execution of the algorithm where the following
	happens. A component $C$ is reassigned to $S_{f(C)}$ and $C$ has the
	following properties: (1) $C$ is a subset of $V_i$ and (2) the vertices in
	$C$ never leave server $S_{f(C)}$ after time $T$. Since each vertex of $V_i$
	is assigned to $S_{f(C)}$ when the algorithm terminates, each vertex of
	$V_i$ is contained in a component with the above properties (when a vertex
	or component is never moved, we set $T=0$). A component $C$ with the above
	properties is a \emph{mixed} component if $C$ contains at least one vertex
	which was not initially assigned to $S_{f(C)}$.  Note that when a mixed
	component $C$ is reassigned to $S_{f(C)}$, $C$ contains at least one vertex
	which was not initially assigned to $S_{f(C)}$ and, hence, $C$ must be moved
	from a server $S_y$, $y \neq f(C)$, to $S_{f(C)}$.

	We bound the cost for mixed components. Let $X$ be the set of
	all mixed components of $V_i$.  Recall that
	Algorithm~\ref{algo:recursive-majority-voting} and \APPROXALGO create the
	same final assignment
	(Lemma~\ref{lem:recursive-majority-voting:same-as-offline}). Hence,
	Algorithm~\ref{algo:recursive-majority-voting} moves the same $\Delta$
	vertices from $V_i$ into $S_{f(V_i)}$ as \APPROXALGO.
	Lemma~\ref{lem:recursive-majority-voting:minority-color} implies that for
	each $C \in X$ at least $|C|/8$ vertices from $C$ are part of the $\Delta$
	vertices moved by \APPROXALGO. Thus, the union of all $C \in X$ contains at
	most $8\Delta$ vertices.

	By Lemma~\ref{lem:recursive-majority-voting:server-cost},
	Algorithm~\ref{algo:recursive-majority-voting} pays at most $O(\alpha |C| \lg |C|)$
	for each $C \in X$ over the entire execution.  Thus, its total cost is
	bounded by
	\begin{align*}
		\sum_{C \in X} O(\alpha |C| \lg |C|)
		\leq O(\alpha \Delta \lg n)
		= O(\ALG_i \cdot \lg n).
	\end{align*}

	Consider the vertices of $V_i$ which are not in mixed
	components.  These vertices must have been part of components in which all
	vertices were originally assigned to $S_{f(V_i)}$.  By
	Lemma~\ref{lem:majority-voting:same-color} (which still applies in the $\ell$
	server setting), these vertices were never moved. Thus, they do not
	incur any cost to the algorithm.
\end{proof}

Next, we show that when the stopping criterion is triggered, the recursive
majority voting algorithm pays $O(n \lg n)$ and cost of the solution by \OFF is
$\Omega(\alpha \varepsilon n / (\ell \lg \ell))$.
\begin{lemma}
\label{lem:recursive-majority-voting:stop}
	When the stopping criterion is triggered, (1) the cost paid by
	Algorithm~\ref{algo:recursive-majority-voting} is $O(\alpha n \lg n)$ and
	(2) the cost paid by $\OFF$ is $\Omega(\alpha \varepsilon n / (\ell \lg \ell))$.
\end{lemma}
\begin{proof}
	Let $Y$ denote the set of all connected components. Part~(1) follows from
	Lemma~\ref{lem:recursive-majority-voting:server-cost} since the total cost
	paid by Algorithm~\ref{algo:recursive-majority-voting} is
	\begin{align*}
		\sum_{C \in Y} O(\alpha |C| \lg |C|)
		\leq \sum_{C \in Y} O(\alpha |C| \lg n)
		= O(\alpha n \lg n).
	\end{align*}

	Now we prove Part~(2). Let $w$ be an internal node of $\T$ with
	children $w_0, w_1$ and suppose w.l.o.g.\ that $w_0$ is overloaded. Since
	the stopping criterion is triggered, $V(w_0)$ contains at least
	$\varepsilon n / (\ell \lg \ell)$ vertices with label $w_1$.

	Let $X$ be the set of all connected components $C$ with the following
	properties: $C$ is assigned to a server in $S(w_0)$ at the time at which the
	stopping criterion is triggered and $C$ contains at least one vertex which
	is labeled with $w_1$.
	
	To show that \OFF performs $\Omega(\varepsilon n / (\ell \lg \ell))$
	vertex moves, we prove that \OFF performs $\Omega(\hw{w_1}{C})$ vertex moves
	for each $C \in X$.  Part~(2) of the lemma follows since the components in
	$X$ contain at least $\varepsilon n /(\ell \lg \ell)$ vertices with label
	$w_1$ and thus
	\begin{align*}
		\OFF
		\geq \sum_{C \in X} \Omega(\alpha \hw{w_1}{C})
		= \Omega(\alpha \varepsilon n / (\ell \lg \ell)).
	\end{align*}
	
	We prove that \OFF moves at least $\Omega(\hw{w_1}{C})$ vertices for
	each $C \in X$ by distinguishing two cases for $C \in X$.
	We define $g$ as the function which maps $C \in X$ to the server it is
	assigned to in the solution of \OFF, i.e., \OFF assigns $C \in X$ to server
	$S_{g(C)}$.

	\emph{Case~1}: $S_{g(C)} \not\in S(w_1)$, i.e., in the final assignment
	of \OFF, the vertices in $C$ are assigned to $S_{g(C)} \not\in S(w_1)$.
	Then \OFF must perform at least $\hw{w_1}{C}$ vertex moves because it
	must move all $w_1$-labeled vertices of $C$ from their initial server
	in $S(w_1)$ to $S_{g(C)} \not\in S(w_1)$.

	\emph{Case~2}: $S_{g(C)} \in S(w_1)$, i.e., in the final solution by
	\OFF, the vertices in $C$ are assigned to a server $S_{g(C)} \in
	S(w_1)$. We show that $C$ contains at least $|C|/4$ vertices without label
	$w_1$. This implies the claim since \OFF must move at least
	$\bhw{w_1}{C} \geq |C|/4$ vertices from servers not in $S(w_1)$ to
	$S_{g(C)} \in S(w_1)$.

	Consider a doubling decomposition $(C_1, \dots, C_k)$ of $C$ (which exists
	by Lemma~\ref{lem:doubling-decomposition-exists}). After $C_1$
	and $C_2$ were merged, the algorithm performed a majority voting step and
	placed $C_1 \cup C_2$ in a server in $S(w_0)$. Thus, $\hw{w_1}{C_1 \cup C_2}
	\leq |C_1 \cup C_2|/2$ (otherwise, the majority voting step would place
	$C_1 \cup C_2$ in a server in $S(w_1)$). Hence,
		$\bhw{w_1}{C_1 \cup C_2}
		= |C_1 \cup C_2| - \hw{w_1}{C_1 \cup C_2}
		\geq |C_1 \cup C_2|/2$.
	Since $|C_1 \cup C_2| \geq |C|/2$, we get $\bhw{w_1}{C} \geq |C|/4$.
\end{proof}

\begin{proof}[Proof of Proposition~\ref{prop:recursive-majority-voting}.]
	The first statement of the proposition is implied by
	Lemmas~\ref{lem:recursive-majority-voting:perfect-partitioning} (perfect
	partitioning),~\ref{lem:recursive-majority-voting:total-cost} (total cost)
	and~\ref{lem:recursive-majority-voting:augmentation} (small augmentation).
	The second statement is proved in
	Lemma~\ref{lem:recursive-majority-voting:stop} (guarantees when stopping
			criterion is triggered).
\end{proof}

\subsection{Small--Large--Rebalance Algorithm for Many Servers}
\label{sec:small-large-rebalance-lservers}

To obtain an efficient algorithm in cases where \OFF moves many vertices, we
reuse the Algorithm~\ref{algo:small-large-rebalance} from
Section~\ref{sec:efficient-small-large-rebalance}.  Note that
Algorithm~\ref{algo:small-large-rebalance} also works with $\ell$ servers
because it did not use the fact that there are only two servers.
In the setting with $\ell$ servers, we obtain the following result.

\begin{proposition}
\label{prop:small-large-rebalance-l-servers}
	Suppose that all servers have capacity $(1+\varepsilon)n/\ell$ for
	$\varepsilon > 0$, i.e., the augmentation is $\varepsilon n / \ell$.  Then the
	cost paid by the more efficient version of
	Algorithm~\ref{algo:small-large-rebalance} is
	$O(\alpha n \lg n + (\OFF \cdot \ell \lg n) / \varepsilon)$.
\end{proposition}
\begin{proof}
	The proof of the lemma is almost the same as the proof of
	Proposition~\ref{prop:small-large-rebalance}. The only difference 
	is that we
	need to bound the number of rebalance operations differently.
	
	The number of vertex moves performed by the algorithm which always moves the
	smaller connected component to the server of the larger connected component
	is $O(n \lg n)$ and, hence, it incurs cost $O(\alpha n \lg n)$. Now, whenever
	a server exceeds its capacity, the algorithm must have moved at least
	$\varepsilon n/\ell$ vertices.  This can only happen $O(\ell \lg n/ \varepsilon)$
	times. By the same arguments as in the proof of
	Proposition~\ref{prop:small-large-rebalance}, each rebalancing operations
	costs $O(\OFF)$. Hence, the cost for all rebalancing steps is bounded by
	$O(\OFF \cdot \ell \lg n / \varepsilon)$.
\end{proof}

We should point out that as in Lemma~\ref{lem:small-large-rebalance}, we could
also do the repartitioning step of Algorithm~\ref{algo:small-large-rebalance} by
taking \emph{any} perfectly balanced assignment respecting the connected
components. In the analysis this would incur $\Theta(n)$ vertex moves for each
such step and, hence, yield an algorithm with $O((n \ell \lg n) / \varepsilon)$
vertex moves in total. However, unlike in the two-server case, finding a
perfectly balanced assignment respecting the connected components is an \NP-hard
problem. Nonetheless, the problem can be solved approximately in polynomial time
at the cost of a constant factor in the competitive ratio. We discuss this in
further detail in Section~\ref{sec:rebalancing:l-servers}.

\subsection{Bringing It All Together: Theorem~\ref{thm:lservers}}
\label{sec:proof-thm-lservers}

\begin{proof}[Proof of Theorem~\ref{thm:lservers}.]
Consider the algorithm which first runs
Algorithm~\ref{algo:recursive-majority-voting} until the stopping criterion is
triggered and then switches to the Algorithm~\ref{algo:small-large-rebalance}
from Section~\ref{sec:small-large-rebalance-lservers}.

If the stopping criterion of the Algorithm~\ref{algo:recursive-majority-voting}
is not triggered, then by Proposition~\ref{prop:recursive-majority-voting} the
cost of the algorithm is $O(\OFF \cdot \lg n)$. Thus, it is $O(\lg n)$-competitive.

If the stopping criterion is triggered, then
Algorithm~\ref{algo:recursive-majority-voting} pays $O(\alpha n \lg n)$ by
Proposition~\ref{prop:recursive-majority-voting}
and the cost of \OFF is $\Omega(\varepsilon n / (\ell \lg \ell))$. Furthermore,
the cost of Algorithm~\ref{algo:small-large-rebalance} is
$O(\alpha n \lg n + (\OFF \cdot \ell \lg n) / \varepsilon)$ by
Proposition~\ref{prop:small-large-rebalance-l-servers}. Hence, we obtain the
following competitive ratio:
\begin{align*}
	&\frac{O(\alpha n \lg n + (\OFF \cdot \ell \lg n) / \varepsilon)}{\OFF}
	= \frac{O(\alpha n \lg n)}{\OFF} + O\left(\frac{\ell \lg n}{\varepsilon}\right) \\
	&\leq O\left(\frac{\alpha n \lg n \cdot \ell \lg \ell}{\alpha \varepsilon n}\right) + O\left(\frac{\ell \lg n}{\varepsilon}\right)
	=  O\left(\frac{\ell \lg n \lg \ell}{\varepsilon}\right).
	\qedhere
\end{align*}
\end{proof}

\section{Distributed and Fast Algorithms}
\label{sec:distributed-fast}

In this section we show how the algorithms from Section~\ref{sec:l-servers} can
be implemented in a distributed setting (Section~\ref{sec:distributed}) and how
they need to be modified to work in polynomial time at the cost of a slightly
worse competitive ratio (Section~\ref{sec:fast}).

We should point out that even though we discuss the distributed and polynomial
time versions of the algorithms separately, they can easily be combined to
obtain a distributed algorithm with polynomial computation time.

\subsection{Distributed Algorithm}
\label{sec:distributed}
While in Section~\ref{sec:l-servers} we presented algorithms in a centralized
model of computation, we now show how
Algorithms~\ref{algo:small-large-rebalance}
and~\ref{algo:recursive-majority-voting} can be implemented in a distributed
model of computation. For realistic parameter settings, the network traffic
caused by our distributed algorithms does not increase (asymptotically) compared
to the traffic caused by moving around the vertices between the servers.

In our distributed model of computation we assume that all servers have access
to: (1) the number of servers $\ell$, (2) the ID of the root server $S_0$, (3) a
shared clock, and (4) all-to-all communication.

When computing the network traffic, we will asymptotically count the number of
messages sent by the algorithms and we further assume that each message contains
$\Theta(\lg n)$ bits. For the sake of simplicity we assume that moving a vertex
from one server to another incurs cost $\alpha = \Theta(\lg n)$.\footnote{Note
	that this is a realistic assumption since in order to move a vertex, a
	server must send the ID of a vertex to another server. Sending the ID of the
	vertex requires $\Theta(\lg n)$ bits.
}
Because of this simplifying assumption we do not have to distinguish between the
number of messages sent by the algorithm and the number of messages used for
moving algorithms.

In this distributed model of computation, we obtain the following main result
for the distributed versions of the algorithms.

\begin{theorem}
\label{thm:distributed}
	Consider a system with $\ell$ servers each of capacity $(1+\varepsilon)n/\ell$
	(i.e., augmentation $\varepsilon n/\ell$) for $\varepsilon \in (0,1/2)$.
	Let $M$ be the number of vertex moves performed by \OFF.

	Then there exists a \emph{distributed}
	$O((\ell \lg n \lg \ell) / \varepsilon)$-competitive algorithm which sends
	\begin{enumerate}
	\item $O(M \lg n)$ messages if $M = O(\varepsilon n / (\ell \log \ell))$,
	\item $O((\ell^2 \lg n) / \varepsilon + n \lg n + (\OFF \cdot \ell \lg n)/\varepsilon)$
		messages if $M = \Omega(\varepsilon n / (\ell \lg \ell))$. 
	\end{enumerate}
	In particular, if $\ell = O(\sqrt{\varepsilon n})$, then the
	algorithm's communication cost does not exceed its cost for moving vertices.
\end{theorem}

We show for Algorithm~\ref{algo:recursive-majority-voting}
(Section~\ref{sec:distributed:recursive-majority-voting}) and for
Algorithm~\ref{algo:small-large-rebalance}
(Section~\ref{sec:distributed:small-large-rebalance}) individually how they can
be implemented distributedly. After that we prove Theorem~\ref{thm:distributed}
in Section~\ref{sec:distributed:proof}.

\subsubsection{Making Algorithm~\ref{algo:recursive-majority-voting} Distributed}
\label{sec:distributed:recursive-majority-voting}
We start by considering the distributed implementation of
Algorithm~\ref{algo:recursive-majority-voting} and obtain the following result. 

\begin{lemma}
\label{lem:recursive-majority-voting:distributed}
	Algorithm~\ref{algo:recursive-majority-voting} can be implemented in a
	distributed model of computation such that the guarantees from
	Proposition~\ref{prop:recursive-majority-voting} still hold. Furthermore, if
	\OFF performs $M$ vertex moves, then we additionally have the following two
	properties:
	\begin{enumerate}
		\item If the stopping criterion is not triggered and the algorithm
		terminates, then the algorithm sent $O(M \lg n)$ messages.
		\item If the stopping criterion was triggered, the algorithm sent
		$O(n \lg n)$ messages.
	\end{enumerate}
\end{lemma}
\begin{proof}
	We start by presenting the necessary modifications to the algorithm and
	analyze the number of sent messages at the end of the proof.

	Let us start by observing that each server can maintain a local
	representation of the bipartition tree $\T$: Since the number of servers
	$\ell$ is known to all servers and $\T$ does not depend on any other
	quantity, each server can compute $\T$ locally. Next, the data structure
	stores for each vertex its ID (requiring $O(\lg n)$ bits) and the ID $j$ of
	the server $S_j$ it was initially assigned to (requiring $O(\lg \ell)$
	bits). Thus, the data structure uses $O(\lg n)$ bits of storage for each
	vertex. In other words, it takes $O(1)$ messages to move a vertex between
	different servers.

	Next, we provide the modifications for checking the stopping criterion,
	small-to-large steps and for majority voting steps.

	Before the algorithm moves a component $C$ from server $S$ to server $S'$,
	$S$ and $S'$ need to check whether the move would trigger the stopping
	criterion.  To do so, $S$ and $S'$ do the following. First, $S$ asks $S'$
	for its ID using $O(1)$ messages. Second, $S$ distinguishes between two
	cases: (1)~$C$ contains at most $\lceil \lg \ell \rceil$ vertices. Then for
	each vertex $v \in C$, $S$ sends a message to $S'$ containing the ID of the
	server $v$ was initially assigned to. This requires $O(|C|)$ messages.
	(2)~$C$ contains more than $\lceil \lg \ell \rceil$ vertices. Then $S$
	locally computes all internal nodes $w$ of the bipartition tree $\T$ which
	contain $S'$ as a leaf. For each such node $w$, let $\bar w$ be the sibling
	of $w$ in $\T$.  Now for each $w$, $S$ computes the number of vertices in
	$C$ which were initially assigned to a server in $S(\bar w)$. Then $S$ sends
	these values to $S'$ using $O(\lg \ell)$ messages. Note that in both cases
	the algorithm does not send more than $O(|C|)$ messages and these messages
	can be charged to the moving cost of $C$ (which requires $\Omega(|C|)$
	messages) which happens after the checking of the stopping criterion.
	Third, $S'$ receives the messages from $S$ and checks locally whether
	receiving $C$ would trigger the stopping criterion. If the stopping
	criterion is not triggered, $S'$ tells $S$ to start moving $C$. If the
	stopping criterion is triggered, $S$ sends a message to the root server
	$S_0$ about this event.  Then $S_0$ informs all other servers about
	switching to Algorithm~\ref{algo:small-large-rebalance}. This requires
	$O(\ell) = O(n \lg n)$ messages.

	Now suppose the algorithm performs a small-to-large step and the stopping
	criterion was previously checked and not triggered. In this case, no
	modifications are necessary: The component $C$ can just be sent from one
	server to the other at the cost of $O(|C|)$ messages (since each vertex in
	$C$ can be sent using $O(1)$ messages).

	Now suppose a server $S$ needs to perform a majority voting step for a
	component $C$. First, observe that $S$ can locally decide whether a majority
	voting step is necessary for $C$ since it must only check the size of $C$.
	Second, when a majority voting step is necessary, $S$ can locally compute
	which server $S'$ will be the recipient of $C$: For each vertex $v \in C$,
	$S$ knows which server $v$ was initially assigned to. Hence, for each $v$,
	$S$ can compute the labels of $v$ w.r.t.\ the bipartitioning scheme from
	Section~\ref{sec:bipartitioning} locally. Since $S$ also knows $\T$, $S$ can
	compute to which server $S'$ the component $C$ should be moved to.  These
	operations do not require any communication between the servers.

	To conclude the proof of the lemma, observe that the distributed algorithm
	performs exactly as many vertex moves as the centralized algorithm.  Hence,
	the guarantees from Proposition~\ref{prop:recursive-majority-voting} still
	hold. Next, we analyze the number of messages sent by the algorithm. A
	small-to-large step moving a component $C$ requires $O(|C|)$ messages.
	Checking the stopping criterion before moving a component $C$ requires
	another $O(|C|)$ messages. Checking whether a majority voting step is
	necessary requires no communication at all. Hence, the number of messages
	used by the algorithm is linear in its number of vertex moves. Thus,
	Proposition~\ref{prop:recursive-majority-voting} implies the two additional
	properties which are claimed in the statement of the lemma.
\end{proof}

\subsubsection{Making Algorithm~\ref{algo:small-large-rebalance} Distributed}
\label{sec:distributed:small-large-rebalance}
For the distributed implementation of Algorithm~\ref{algo:small-large-rebalance}
we obtain the following result.

\begin{lemma}
\label{lem:small-large-rebalance:distributed}
	Algorithm~\ref{algo:small-large-rebalance} can be implemented in a
	distributed model of computation such that the guarantees from
	Proposition~\ref{prop:small-large-rebalance-l-servers} still hold.
	Furthermore, if \OFF performs $M$ vertex moves, then the algorithm sends at
	most $O((\ell^2 \lg n)/\varepsilon + n \lg n + (M \ell \lg n)/\varepsilon)$
	messages.
\end{lemma}
\begin{proof}
	We start by stating which modifications need to be made to make
	Algorithm~\ref{algo:small-large-rebalance} distributed.

	First, suppose that Algorithm~\ref{algo:small-large-rebalance} performs a
	small-to-large step moving a component $C$ and that this move does not make
	any server exceed its capacity. In this case, no modifications are necessary
	and the number of messages sent is $O(|C|)$ as we have seen in the proof of
	Lemma~\ref{lem:recursive-majority-voting:distributed}.

	Second, suppose that a small-to-large step wants to move component $C$ to
	server $S$ which would cause $S$ to exceed its capacity. Then the algorithm
	performs the following operations:
	\begin{enumerate}
		\item\label{step:rebuild-required}
		$S$ informs the root server $S_0$ that a rebuild is required.
		\item\label{step:send-edges}
		$S_0$ asks all $\ell$ servers to send the edges that were inserted
		\emph{and caused the merge of two connected components} since the last
		rebuild. The servers send of all these edges together with the
		timestamps when they were inserted.
		\item\label{step:simulate}
		$S_0$ locally simulates the whole system from the beginning and obtains knowledge
		about all connected components and which servers they are assigned to.
		\item\label{step:move-components}
		$S_0$ tells all other servers $S_j$ which components need to be moved
		and all servers perform the necessary moves.
	\end{enumerate}

	Since the distributed algorithm performs exactly the same vertex moves as
	the centralized algorithm, the distributed algorithm is correct and provides
	the same guarantees as provided in
	Proposition~\ref{prop:small-large-rebalance-l-servers}. We only need to
	analyze how many messages the algorithm sends. To do so, we analyze each
	step separately.

	Every time Step~\ref{step:rebuild-required} is performed, it requires $O(1)$
	messages. As there are $O((\ell \lg n)/\varepsilon)$ rebuilds in total,
	Step~\ref{step:rebuild-required} sends $O((\ell \lg n)/\varepsilon)$
	messages in total.

	To bound the number of messages sent in Step~\ref{step:send-edges}, recall
	that in total there are only $O(n)$ edges which merge connected components.
	Hence, sending these edges requires $O(n)$ messages. Furthermore, when a
	server did not obtain an edge merging two connected components between two
	rebuilds, it can inform $S_0$ about this in $O(1)$ messages. As this can be
	the case for at most $\ell$ servers and since there are
	$O((\ell \lg n)/\varepsilon)$ rebuilds, at most
	$O((\ell^2 \lg n)/\varepsilon)$ messages are sent when servers did not
	receive new edges.

	In Step~\ref{step:simulate}, $S_0$ locally simulates the system. This does
	not incur any network traffic.

	Now consider Step~\ref{step:move-components}. During a rebuild, the number
	of components which the algorithm needs to reassign is trivially bounded by
	the number of vertex moves performed during the rebuild.  Thus,
	Proposition~\ref{prop:small-large-rebalance-l-servers} implies that only
	$O(n \lg n + (\OFF \cdot \ell \lg n)/\varepsilon)$ messages are required for all
	invocations of Step~\ref{step:move-components}.

	In total, we obtain that the algorithm sends at most
	$O((\ell^2 \lg n)/\varepsilon) + n \lg n + (M \ell \lg n)/\varepsilon)$
	messages, where $M$ is the number of vertices moved by \OFF.
\end{proof}

\subsubsection{Proof of Theorem~\ref{thm:distributed}}
\label{sec:distributed:proof}
	To prove the claim about the competitive ratio of the algorithm observe that
	the distributed algorithm performs exactly the same vertex moves as the
	centralized algorithm. Hence, the cost paid by both algorithms is the same
	and the distributed algorithm has the same competitive ratio as the
	centralized algorithm in Theorem~\ref{thm:lservers}. 
	
	The claim about the number of messages sent by the algorithm follows from
	Lemma~\ref{lem:recursive-majority-voting:distributed} and
	Lemma~\ref{lem:small-large-rebalance:distributed} and summing over the
	number of messages.
	
	To prove the last claim of the theorem, we distinguish two cases. If the
	stopping criterion was not triggered, then the claim holds by
	Lemma~\ref{lem:recursive-majority-voting:distributed}. If the stopping
	criterion was triggered, then if $\ell = O(\sqrt{\varepsilon n})$, we obtain
	that the total number of messages is
	\begin{align*}
		O((\ell^2 \lg n)/\varepsilon) + n \lg n + (M \ell \lg n)/\varepsilon)
		&= O( (\varepsilon n \lg n)/\varepsilon + n \lg n + (M \ell \lg n)/\varepsilon ) \\
		&= O( n \lg n + (M \ell \lg n)/\varepsilon ),
	\end{align*}
	which is exactly the number of vertices moved by
	Algorithm~\ref{algo:small-large-rebalance}.
\qed

\subsection{Fast Algorithms}
\label{sec:fast}

In this section, we discuss the computational challenges when computing
perfectly balanced assignments. These computational problems occur when
Algorithm~\ref{algo:small-large-rebalance} performs rebalancing steps (see
Section~\ref{sec:small-large-rebalance} and
Section~\ref{sec:small-large-rebalance-lservers}). So far, we were only
concerned with algorithms which try to minimize the vertex moves while using
potentially exponential running time. We now consider polynomial time
algorithms. The only step where our algorithms might use exponential time is
during rebalancing. Thus we show next how to perform the rebalancing operations
in polynomial time. In the case of $\ell > 2$ servers, our polynomial time
algorithms perform slightly more vertex moves than the exponential time
algorithms.

We discuss the two server case which can be solved optimally in polynomial time
in Section~\ref{sec:rebalancing:two-servers}. In
Section~\ref{sec:rebalancing:l-servers}, we argue that in the general case with
$\ell > 2$ servers this problem is \NP-hard. We resolve this issue in
Section~\ref{sec:rebalancing:l-servers-approx} by computing approximately
balanced assignments in polynomial time.

\subsubsection{Computing Perfectly Balanced Assignments for Two Servers}
\label{sec:rebalancing:two-servers}
We consider computing a perfectly balanced assignment respecting the connected
components for two servers.  Specifically, we provide a dynamic program which
can find such an assignment in polynomial time.

The dynamic program works as follows. Suppose $C_1, \dots, C_q$ are the
connected components assigned to the two servers. Now let $k_i = |C_i|$ for
$i=1,\dots,q$. We create a set $\S$ consisting of integers with the following
property: Each number $s \in \S$ corresponds to a set of connected components
$\C$ such that $|\bigcup_{C \in \C} C| = s$. That is, whenever $s \in \S$, there
exists a set of connected components which together contain $s$ vertices.  For
each $s\in\S$, the algorithm maintains a set of connected components explicitly.
We denote the components corresponding to value $s\in\S$ by $\comp(S)$.  

At the beginning of a rebalancing step, the algorithm sets $\S = \{0\}$. The
connected component corresponding to value $0$ is simply the empty set of
vertices, i.e., $\comp(0) = \emptyset$.  For $i = 1,\dots,q$ the algorithm does
the following. Iterate over all $s \in \S$ and over all components and add
$s + k_i$ to $\S$ if $s + k_i \not\in \S$ and $C_i \not\in \comp(S)$.
Whenever a new value $s + k_i$ is added to $\S$, set
$\comp(s + k_i) = \comp(s) \cup \{C_i\}$.

As soon as the value $n/2$ is added to $\S$, the dynamic program stops and
assigns all vertices in $\comp(n/2)$ to the left server and all remaining
vertices to the right server.

The correctness of the above algorithm is clear by construction. We only need to
show that it finishes in polynomial time.

Note that the above dynamic program runs in time $O(q |\S|)$. Now observe that $q$
is bounded by $n$ since there are at most $n$ connected components. Furthermore,
for each subset $\C \subseteq \{C_1, \dots, C_q\}$, we have that
$\sum_{C\in\C} |C| \leq n$ (because the components in $\C$ cannot contain more
than $n$ vertices). Thus, $|\S| \leq n+1$ because each value $s\in\S$
corresponds to a subset of components $\C \subseteq \{C_1,\dots,C_q\}$ and each
value $s \in \{0,\dots,n\}$ is only added once to $\S$.  Hence, the algorithm
runs in time $O(q |\S|) = O(n^2)$.

\subsubsection{Computing Perfectly Balanced Assignments for Many Servers}
\label{sec:rebalancing:l-servers}
We consider computing a perfectly balanced assignment respecting the connected
components for $\ell$ servers.

Let $C_1, \dots, C_q$ be the connected components which are assigned to the
$\ell$ servers. To find a perfectly balanced assignment respecting the connected
components, we need to find a partition of the set $\S = \{|C_1|,\dots,|C_q|\}$
into $\ell$ subsets $\S_1, \dots, \S_\ell$ such that for each subset $\S_i$ we
have that $\sum_{s \in \S_i} s = n/\ell$.

Unfortunately, the above problem is known to be \NP-complete, see, e.g.,
the result about multi-processor scheduling in Garey and
Johnson~\cite{garey79computers}. However, since we prove our results in
the online model of computation, which allows unlimited computational power, the
algorithm can solve this \NP-complete problem.
We note that this problem has also been
studied in practice, see, e.g., Schreiber et al.~\cite{ethan18optimal} and
references therein.

See Section~\ref{sec:rebalancing:l-servers-approx} for how this problem can be
solved approximately at the cost of a constant in the competitive ratio of the
algorithm.

\subsubsection{Computing Approximately Balanced Assignments for Many Servers}
\label{sec:rebalancing:l-servers-approx}

Previously we have we seen that \emph{perfectly} balanced assignments for $\ell$
servers cannot be computed in polynomial time unless $\P = \NP$
(Section~\ref{sec:rebalancing:l-servers}). Thus, we now consider computing
\emph{approximately} balanced assignments for $\ell$ servers which is sufficient
for our purpose: Let $\varepsilon' > 0$ be a constant.  An assignment is
\emph{$(1+\varepsilon')$-approximately balanced} if each server has load at most
$(1+\varepsilon') n / \ell$. Using this definition, we obtain the following
result.

\begin{proposition}
\label{prop:approx-balanced-assignment}
	Let $\varepsilon > \varepsilon' > 0$ be constants and suppose each server
	has capacity $(1+\varepsilon)n/\ell$. Then a
	$(1+\varepsilon')$-approximately balanced assignment for $\ell$ servers can
	be computed in polynomial time.
\end{proposition}

Using the proposition (which we prove at the end of the subsection), we obtain a
polynomial time algorithm with a slightly worse competitive ratio than that of
Theorem~\ref{thm:lservers}.
\begin{theorem}
\label{thm:poly-time-algo}
	Given a system with $\ell$ servers each of capacity $(1+\varepsilon)n/\ell$,
	for constant $\varepsilon \in (0,1/2)$, then there exists an
	$O((\ell^2 \lg n \lg \ell) / \varepsilon^2)$-competitive algorithm which
	runs in polynomial time.
\end{theorem}
\begin{proof}
	First, observe that Algorithm~\ref{algo:recursive-majority-voting} runs in
	polynomial time. Thus, the result of
	Proposition~\ref{prop:recursive-majority-voting} also holds for polynomial
	time algorithms.

	Second, consider a modification of
	Algorithm~\ref{algo:small-large-rebalance} where at each rebalancing step we
	compute a $(1+\varepsilon')$-approximately balanced assignment for
	$\varepsilon' = \varepsilon / 2$. Such an approximately balanced assignment
	can be computed in polynomial time due to
	Proposition~\ref{prop:approx-balanced-assignment}. Thus, the modified
	algorithm runs in polynomial time.

	Observe that now all steps of the resulting algorithm can be computed in
	polynomial time.  It is left to bound the competitive ratio of the modified
	algorithm.
	
	We start by bounding the cost paid by the modified version of
	Algorithm~\ref{algo:small-large-rebalance}.
	Note that each approximate rebalancing step incurs cost at most $O(\alpha n)$;
	recall that $\alpha$ denotes the cost for moving a vertex to a different
	server.  Now we bound the number of approximate rebalancing steps.  Recall
	from Lemma~\ref{lem:small-large-only} that the number of vertex moves due to
	small-to-large steps is at most $O(n \lg n)$. Now whenever a new
	approximately balanced assignment is computed, the small-to-large steps must
	have moved at least $\Omega( (\varepsilon - \varepsilon') n/\ell)$ vertices
	to exceed the capacity of one of the servers. Thus, the total number of
	approximate rebalancing operations is bounded by
	$O((\ell \lg n) / (\varepsilon - \varepsilon'))$ and, hence, the total cost
	of Algorithm~\ref{algo:small-large-rebalance} with approximate rebalancing
	steps is bounded by
	$O((\alpha n \ell \lg n) / (\varepsilon - \varepsilon'))$.

	Altogether, we obtain the following competitive ratio by following the steps
	from the proof of Theorem~\ref{thm:lservers}
	(Section~\ref{sec:proof-thm-lservers}):
	\begin{align*}
		\frac{O(\alpha n \lg n + (\alpha n \ell \lg n) / (\varepsilon - \varepsilon'))}{\OFF}
		= O\left( \frac{ \alpha n \lg n + (\alpha n \ell \lg n) / (\varepsilon - \varepsilon') }
				{ \alpha \varepsilon n / (\ell \lg \ell) } \right)
		=  O\left(\frac{\ell^2 \lg n \lg \ell}{\varepsilon^2}\right),
	\end{align*}
	where in the last step we used that $\varepsilon - \varepsilon' = \varepsilon/2$.
\end{proof}

To prove Proposition~\ref{prop:approx-balanced-assignment}, we consider the
makespan minimization problem in which there are $k$ jobs with processing times
$p_1, \dots, p_k$ which must be assigned to $\ell$ identical machines. Given an
assignment of the jobs to the machines, the maximum running time time of any
machine is called the \emph{makespan}. The goal is to find an assignment of the
jobs to the machines which minimizes the makespan.

The makespan minimization problem is known to be \NP-hard but Hochbaum and
Shmoys~\cite{hochbaum87using} presented a polynomial time approximation scheme
(PTAS).
\begin{lemma}[Hochbaum and Shmoys~\cite{hochbaum87using}]
\label{lem:makespan}
	Let $\varepsilon' > 0$ be a constant. Then there exists an algorithm which
	computes a $(1+\varepsilon')$-approximate solution for the makespan
	minimization problem in polynomial time.
\end{lemma}

Using the result from the lemma we can prove
Proposition~\ref{prop:approx-balanced-assignment}.
\begin{proof}[Proof of Proposition~\ref{prop:approx-balanced-assignment}]
	Suppose the system currently contains connected components $C_1,\dots,C_k$.
	We consider these connected components as the jobs of the makespan
	minimization problem with processing times $p_i = |C_i|$ for $i=1,\dots,k$.
	The machines correspond to the $\ell$ servers.

	Note that the optimal solution for the instance of the makespan minimization
	problem is $n/\ell$: Since we have made the assumption that in the final
	assignment all servers have load exactly $n/\ell$, there must exist a
	perfectly balanced assignment from the components $C_i$ to the servers
	$S_i$. In other words, there exists an assignment of the jobs to the
	machines such that each machine has running time $n/\ell$ and, hence, the
	optimal makespan is $n/\ell$.

	By running the algorithm from Lemma~\ref{lem:makespan}, we obtain a
	$(1+\varepsilon')$-approximate solution for the makespan minimization
	problem. Since the optimal solution for this problem is $n/\ell$, each 
	machine has load at most $(1+\varepsilon')n/\ell$ in the solution returned
	by the algorithm from Lemma~\ref{lem:makespan}. Assigning the components
	$C_i$ to the servers in exactly the same way as the corresponding jobs are
	assigned to the corresponding machines, we obtain a
	$(1+\varepsilon')$-approximately balanced assignment in polynomial time.
\end{proof}

\section{Lower Bounds}
\label{sec:lowerbounds}

To study the optimality of our algorithms, we derive bounds on the competitive
ratios which can be achieved by \emph{any} deterministic online algorithm. 

The following theorem provides a lower bound of $\Omega(1/\varepsilon + \lg n)$.
The lower bound has the following two main consequences: (1) If an algorithm is
only allowed to use constant augmentation (i.e., servers of capacity $n/\ell + O(1)$),
then the lower bound implies that any algorithm must have a competitive ratio of
$\Omega(n)$.\footnote{To obtain servers of capacity $n/\ell + O(1)$, we must set
	$\varepsilon = O(1)/n$.}
(2) The lower bound holds even in the setting in which there are only two
servers. Thus, the algorithm from Section~\ref{sec:two-servers} for the two
server setting is close to optimal (up to a $O(\min\{ 1/\varepsilon, \lg n\})$ factor) and the
generalized algorithm from Section~\ref{sec:l-servers} is optimal up to a
$O(\ell \lg \ell \min\{1/\varepsilon, \lg n\})$ factor.

\begin{theorem}
\label{thm:det-lower-bound}
	Suppose there are two servers of capacity $(1+\varepsilon)n/2$ for
	$\varepsilon \leq 0.98$. Then any deterministic online algorithm
	must have a competitive ratio of $\Omega(1/\varepsilon + \lg n)$.
\end{theorem}

To prove the theorem, we show in Section~\ref{sec:connected-components}
that there exist input sequences such that \emph{either} an algorithm always assigns
vertices of the same connected component to the same server \emph{or} it has
prohibitively high cost. Using this fact,
we prove our concrete lower bounds in Section~\ref{sec:concrete-lower-bounds}.

\subsection{Assigning Connected Components to Servers}
\label{sec:connected-components}

In this subsection, we give an important reduction 
which will be useful to derive the lower
bounds in the next subsection (Section~\ref{sec:concrete-lower-bounds}).  
This reduction lets us assume that every competitive algorithm will always assign
vertices of the same connected component to the same server.
   
More concretely, we show that every sequence of edges $\sigma$ can be manipulated
to a new edge sequence $\sigma'$ such that: (1) $\sigma$ reveals the same edges
as $\sigma'$ and (2) on input $\sigma'$, every algorithm \emph{either} moves the
vertices of the same connected components to the same server, \emph{or} has
prohibitively high cost and, hence, cannot be competitive.

We first prove the following technical lemma.

\begin{lemma}
\label{lem:connected-components}
	Consider a sequence $\sigma$ which reveals the edges $\emptyset \neq E^* \subseteq E$.
	Let $C_1, \dots, C_q$ be the connected components induced by $E^*$.

	Then for each initial assignment there exists an input sequence $\sigma'$
	consisting only of edges in $E^*$ such that either (1) at some point during
	the input sequence the algorithm assigns all vertices from each $C_i$ to the
	same $S_j$ or (2) the cost of the algorithm is at least
	$\Omega(\alpha n^3)$.
\end{lemma}
\begin{proof}
	We will construct an input sequence $\sigma'$ provided by the adversary such
	that either Property~(1) or Property~(2) must hold.

	Consider an arbitrary initial assignment and pick the ground truth
	components $V_i$ such that they do not coincide with the initial assignment
	of the vertices to the servers, i.e., $V_i \neq \init(S_j)$ for all $i, j$.
	Let $E^* = \{e_1', \dots, e_t' \}$ be the edges revealed by the adversary
	and suppose that $E^*$ contains at least one edge $(u,v)$ such that $u$ and
	$v$ are assigned to different servers in the initial assignment.

	Now consider the input sequence $\sigma' = (e_1, \dots, e_r)$ with
	$r = \lceil \alpha n^3 t \rceil$ which consists of the edges
	$(e_1', \dots, e_t')$ in $E^*$ concatenated $\lceil \alpha n^3 \rceil$
	times.

	Suppose that while running the algorithm there always exists a $C_i$ such
	that not all vertices from $C_i$ are assigned to the same server $S_j$,
	i.e., Claim~(1) does not apply. We show that then Claim~(2) must apply.

	Consider the state of the algorithm prior to a single subsequence containing
	the edges $(e_1', \dots, e_t')$. By assumption at least one edge $e_i'$ must
	be between two vertices from different servers. Now the algorithm must
	either pay $1$ for communication along this edge or it must move one of the
	edge's endpoints at the cost of $\alpha$ to avoid paying for communication
	along this edge. Thus, the algorithm must pay at least $\Omega(1)$ for the
	subsequence $(e_1', \dots, e_t')$.

	As there are $\lceil \alpha n^3 \rceil$ such subsequences, the algorithm
	must pay at least $\Omega(\alpha n^3)$ in total.
\end{proof}

As we will see, the lemma essentially allows us to assume that 
every algorithm which obtains an edge between vertices on different
clusters, must move their connected components to the same cluster.
That is, given an input sequence $\sigma$, in our lower bound proof,
we can employ Lemma~\ref{lem:connected-components} to obtain an input 
sequence $\sigma'$ which
does not reveal any additional edges and which forces every algorithm to have
Property~(1) or Property~(2).

Now observe that if an algorithm has Property~(2), since the cost of \OFF
are always bounded by $O(\alpha n)$ (\OFF moves each vertex at most once), the
algorithm cannot be competitive: the competitive ratio must 
be at least $\Omega(n^2)$, much higher than the competitive ratios derived 
in this paper.
Hence, in the following we can assume that every algorithm with a competitive
ratio better than $\Omega(n^2)$ must satisfy Property~(1) of
Lemma~\ref{lem:connected-components}.

\subsection{Lower Bound Proofs}
\label{sec:concrete-lower-bounds}
In this subsection, we prove Theorem~\ref{thm:det-lower-bound} by proving two
different lower bounds: The first lower bound asserts a competitive ratio of
$\Omega(1/\varepsilon)$ and the second lower bounds asserts a competitive ratio
of $\Omega(\lg n)$.

In the lower bound constructions we heavily exploit that we provide hard
instances against \emph{deterministic} algorithms, i.e., we will rely on the
fact that at each point in time the adversary knows exactly which assignment the
online algorithm created.

Furthermore, we assume that after each edge which was provided by the adversary,
the algorithm creates an assignment such that all vertices of the same
connected component are assigned to the same server. This assumption is
admissible by the discussion in Section~\ref{sec:connected-components}.

We start by proving the lower bound of $\Omega(1/\varepsilon)$. 

\begin{lemma}
	Consider the setting with two servers which both have capacity
	$(1+\varepsilon)n/2$ for $\varepsilon > 0$.
	
	Then for each deterministic online algorithm \ALG there exists an input
	sequence $\sigma$ such that the cost of \ALG is $\Omega(\alpha n)$ and the
	cost paid by \OFF is $O(\alpha \varepsilon n)$. Thus, the competitive ratio
	of every online algorithm is $\Omega(1/\varepsilon)$.
\end{lemma}
\begin{proof}
	Choose an arbitrary initial assignment of $n$ vertices to the $\ell$
	servers.  Let $K = \varepsilon n / 2$ denote the allowed augmentation of the
	servers. The initial assignment is as follows. In the left server, there are
	$q = n / (2 (K+1))$ connected components $C_1, \dots, C_q$ of size $K+1$.  On
	the right server, we build one connected component of size $K+1$ denoted $C$
	and one large connected component of size $n-K-1$ denoted $C'$. First, the
	adversary provides all edges of these connected components at no cost to the
	algorithm.

	Then the adversary inserts an edge from a vertex in $C_1$ to a vertex in $C$.
	Since $C_1$ has size $K+1$ and the right server currently has $n/2$ vertices,
	the algorithm cannot move $C_1$ to the right server. For the same reason, the
	algorithm cannot move $C$ to the left server either. Thus, the algorithm's
	only option to bring $C_1$ and $C$ to the same server is to replace $C$ with
	some $C_i$ at the cost of $2 \alpha (K+1)$.

	We will refer to the merged connected component of size $2(K+1)$ as $D$.
	Note that $D$ must be on the left server. Now let $C_i$ be the connected
	component of size $K+1$ on the right server. The adversary adds an edge from a
	vertex in $D$ to a vertex in $C_i$. By the same reasoning as before, the
	algorithm must now pick some $C_j$, $j \neq i$, of size $K+1$ from the left
	server and swap it with $C_i$. This costs another $2 \alpha (K+1)$.

	The adversary continues the previous procedure until only a $C_i$ of size
	$K+1$ is left on the left server and then she connects $C_i$ and $C'$. This
	gives the final partitioning of the vertices.
	
	We observe that each vertex which is 
	on the left server at the very end, has
	been on the right server exactly once 
	during the execution of the algorithm.
	Thus, the costs paid by the algorithm must be $\Omega(\alpha n)$.

	Note that \OFF pays exactly $\alpha( K+1 )$ because it can determine
	beforehand which $C_i$ must be moved to the right server and only move that
	connected component. Before, we have seen
	that any deterministic algorithm must pay at least $\Omega(\alpha n)$. Thus,
	the competitive ratio is $\Omega(n / K)$.
\end{proof}

Next, we prove the $\Omega(\lg n)$ lower bound for the competitive ratio of
deterministic algorithms.
\begin{lemma}
	Consider the setting with two servers which both have capacity
	$(1+\varepsilon)n/2$ for $\varepsilon \leq 0.98$.
	
	Then for each deterministic online algorithm \ALG there exists an input
	sequence $\sigma$ such that the cost paid by \ALG is $\Omega(\alpha n \lg n)$
	and the cost paid by \OFF is $O(\alpha n)$. Thus, the competitive ratio of every
	deterministic online algorithm is $\Omega(\lg n)$.
\end{lemma}
\begin{proof}
	Choose an arbitrary initial assignment of $n$ vertices to the $\ell$
	servers. Since we want to prove a lower bound, we can assume that $n$ is a
	power of 2. Thus, suppose that $n = 2^a$ for $a \geq 1000$.
	
	In our hard instance, we are creating a sequence of edge insertions which
	proceeds in $\Theta(\lg n)$ rounds. When round $i$ starts, all connected
	components have size $2^i$ induced by the previously provided edges, and
	when round $i$ finishes, all connected components have size $2^{i+1}$. We
	show that \ALG pays $\Omega(\alpha n)$ in each round. This
	implies the claimed cost of $\Omega(\alpha n \lg n)$ for \ALG.  The cost for
	\OFF follows immediately from Lemma~\ref{lem:cost-opt-2server} which states
	that \OFF never pays more than $O(n)$ when there are only two servers.

	When \ALG starts and no edge was provided by the adversary, all connected
	components have size $1 = 2^0$, i.e., the connected components are isolated
	vertices.

	Now suppose round $i = 0, \dots, \lg n$ starts. By induction, all connected
	components have size $2^i$. We now define a sequence of edge insertions for
	round $i$ which forces \ALG to pay $\Omega(\alpha n)$ and after which all
	connected components have size $2^{i+1}$.

	Let $z$ denote the current number of connected components of size $2^i$.
	When round $i$ starts, there are exactly $z = n/2^i = 2^{a-i}$ connected
	components of size $2^i$ each.  Recall that each server has capacity
	$(1+\varepsilon)n/2$.  Thus, at most
	\begin{align*}
		y_i = (1+\varepsilon)n/2^{i+1} \leq 1.98 \cdot 2^{a-i-1}
	\end{align*}
	connected components of size $2^i$ can be assigned to each server.

	Now suppose there exists an edge $(u,v)$ such that $C_u$ and $C_v$ are of
	size $2^i$ and they are assigned to different servers; we call such an edge
	\emph{expensive}. When the adversary inserts an expensive edge, \ALG must
	pay $\Omega(\alpha 2^i)$ for moving $C_u$ or $C_v$ to a different server.

	The strategy of the adversary is to insert expensive edges as long as they
	exist. Once no expensive edges exist anymore, the adversary connects all
	remaining components of size $2^i$ arbitrarily until all components have
	size $2^{i+1}$.

	Note that expensive edges exist as long as $z > y_i$ (because when this
	inequality is satisfied, not all connected components of size $2^i$ can be
	assigned to the same server). Furthermore, observe that when the adversary
	inserts an expensive edge, $z$ decreases by 2.
	
	Now we prove a lower bound on the number of expensive edges $p$. By the
	previous arguments, $p$ must be large enough such that:
	\begin{align*}
		z = 2^{a-i} - 2p \leq y_i.
	\end{align*}
	Solving this inequality for $p$, we obtain
	\begin{align*}
		p
		&\geq 2^{a-i-1} - 1.98 \cdot 2^{a-i-2} \\
		&= 2^{a-i-1} (1 - 0.99) \\
		&= 0.01 \cdot 2^{a-i-1}.
	\end{align*}

	We conclude that that adversary can perform $\Omega(2^{a-i})$ expensive edge
	insertions.  Since for each of these edge insertions, \ALG must pay
	$\Omega(\alpha 2^i)$, we obtain that the cost paid by \ALG in round $i$ is
	\begin{align*}
		\Omega(\alpha \cdot 2^{a-i} \cdot 2^i)
		= \Omega(\alpha \cdot 2^{a-2})
		= \Omega(\alpha n).
		& \qedhere
	\end{align*}
\end{proof}

\section{Sample Applications: A Distributed Union Find Algorithm and Online
	$k$-Way Partitioning}
\label{sec:applications}

In this section we provide two sample applications for our model and our
algorithms. First, we show that our results can be used to solve a distributed
union find problem and we give an example where a union find data structure is
used in practice. Second, we show that our algorithms imply competitive
algorithms for an online version of the $k$-way partitioning problem.

\subsection{Distributed Union Find}
\label{sec:union-find}
Recall that in the \emph{static} union find problem, there are $n$ elements
from a universe $\U$ and initially there are $n$ sets containing one element
each. The data structure supports two operations: $\union(u,v)$ and $\find(u)$.
Given two elements $u, v \in \U$, the operation $\union(u,v)$ merges the sets
containing $u$ and $v$. The operation $\find(u)$ returns the set containing $u$.

In the distributed setting we consider, elements are stored across $\ell$
servers. Each server has enough capacity to store $(1+\varepsilon) n / \ell$
elements and we have the natural constraint that elements from the same set must
always be stored on the same server (in order to maintain locality for elements
from the same set). We consider a setting in which all sets have size $n/\ell$
when the algorithm finishes.

Note that if the sets of $u, v \in \U$ are stored on different servers when the
operation $\union(u,v)$ is performed, one of the sets containing $u$ or $v$ must
be moved to a different server. The goal of an algorithm is to minimize the
moving cost caused by $\union$-operations. 

When analyzing the moving cost, we will compare with an optimal offline
algorithm which knows in advance which $\union$-operations will be performed.
Thus, the optimal algorithm can move from the initial assignment to the final
assignment at the minimum possible cost. For our analysis we will compute the
competitive ratio between an online algorithm solving the above problem and the
optimal offline algorithm (as also detailed in Section~\ref{sec:preliminaries}).

Using the algorithms from Sections~\ref{sec:l-servers}
and~\ref{sec:distributed}, we obtain the following result.
\begin{theorem}
\label{thm:union-find}
	Consider a system with $\ell$ servers each of capacity
	$(1+\varepsilon)n/\ell$ for $\varepsilon \in (0,1/2)$.	Then there exists a
	\emph{distributed} $O((\ell \lg n \lg \ell) / \varepsilon)$-competitive
	algorithm for the distributed union find problem. Moreover, for $\ell =
	O(\sqrt{\varepsilon n})$ servers, the algorithm's communication cost does
	not exceed its cost for moving vertices.
\end{theorem}
\begin{proof}
	The theorem follows immediately from Theorem~\ref{thm:distributed} by the
	following reduction from the model in Section~\ref{sec:preliminaries}. We
	identify vertices in the model from Section~\ref{sec:preliminaries} with
	elements from the universe $\U$ in the union find model. Furthermore, for
	each operation $\union(u,v)$ we insert an edge $(u,v)$ into the model from
	Section~\ref{sec:preliminaries}.  Since all algorithms we considered always
	collocate vertices from the same connected component, they satisfy the
	constraint that elements from the same set must be assigned to the same
	server. Moreover, in our analysis we were able to focus on the number of
	vertex moves due to Lemma~\ref{lem:reduction}. In our proofs, we showed
	competitive bounds for the number of vertex moves performed by the algorithm
	from Theorem~\ref{thm:distributed} compared with an optimal offline
	algorithm. Thus, the same bounds as derived in Theorem~\ref{thm:distributed}
	apply.
\end{proof}

For $\ell = \Omega(\sqrt{\varepsilon n})$ servers and the exact number of
messages sent by the algorithm, see Theorem~\ref{thm:distributed}. The
guarantees from Theorem~\ref{thm:distributed} carry over immediately.

An examples where distributed union find data structures are used in practice is
search engines~\cite{broder97syntactic}.  A search engine stores many different
documents from the Web over multiple servers. Now union find data structures are
used to collocate duplicate documents on the same server, i.e., when documents
$u$ and $v$ are identified as duplicates the operation $\union(u,v)$ is used to
collocate these documents (and all previously identified duplicates) on the same
server.  Furthermore, union find data structures are used to find blocks in
dense linear systems and in pattern recognition tasks (see Cybenko et
al.~\cite{cybenko88practical} and references therein).

\subsection{Online $k$-Way Partitioning}
\label{sec:partition}
The model and algorithms we study in this paper can also be used to solve an
online variant of the $k$-way partition problem~\cite{ethan18optimal}. In the
static version of the $k$-way partition problem one is given a (multi-)set of
integers $\S$ and the task is to partition $\S$ into $k$ subsets
$\S_1,\dots,\S_k$ such that the sum of all subsets is (approximately) equal.

Our model and our algorithms can be used to solve the following online version
of this fundamental problem. Initially, $\S$ contains $n$ integers and all
integers are $1$. Each integer is assigned to one of $\ell$ bins and each bin has
capacity $(1+\varepsilon) n/\ell$.  Now in an online sequence of operations, an
adversary picks two integers from $\S$ and these integers are added. For
example, after adding integers $a, b \in \S$, $\S$ becomes
$\S = (\S \cup \{a+b\}) \setminus \{a,b\}$. During this sequence of operations
an online algorithm must ensure that the load of all bins is always bounded by
$(1+\varepsilon) n/\ell$. We work under the assumption that after each operation
there always exists an assignment from the integers in $\S$ to the bins such
that each bin has load exactly $n/\ell$. We further assume that at the end of
the sequence of operations there are $\ell$ integers and each integer is
$n/\ell$.

Note that when two integers $a, b \in \S$ from different bins are added, either
$a$ or $b$ must be moved to a different bin. This might cause that bin to exceed
its capacity. 

We will analyze algorithms which have small moving cost. That is, the cost of an
algorithm is the sum of the numbers it has moved. We consider the competitive
analysis of online algorithms compared with an optimal offline algorithm which
knows the sequence of additions in advance and which can move the numbers at
optimal cost.

We then obtain the following result for the $k$-way partitioning problem.
\begin{theorem}
	Consider a system with $\ell$ bins each of capacity $(1+\varepsilon)n/\ell$
	for $\varepsilon \in (0,1/2)$.	Then there exists a
	$O((\ell \lg n \lg \ell) / \varepsilon)$-competitive algorithm for the
	$k$-way partition problem.
\end{theorem}
\begin{proof}
	We can relate the online version of the $k$-way partition problem to the
	model we study by identifying integers and the sizes of connected
	components. Initially, we identify each $s \in \S$ with a single vertex.
	Note that this can be done since initially $s = 1$ and thus $s$ and the size
	of its corresponding connected component are the same. After that, when two
	integers $a$ and $b$ are added, we take their corresponding connected
	components $C_a$ and $C_b$ and insert an edge between them. Note that the
	resulting integer $a+b$ corresponds to the connected component $C_a \cup
	C_b$ and their sizes agree, i.e., $a + b = |C_a \cup C_b|$. Now observe that
	summing the moving cost for integers is the same as counting the number of
	vertex reassignments for connected components. Thus, the result of the
	theorem follows from Theorem~\ref{thm:lservers}.
\end{proof}

\section{Related Work}
\label{sec:related}

The design of more flexible networked systems
that can adapt to their workloads has received
much attention over the last years, 
with applications for 
traffic engineering~\cite{swan,b4}, 
load-balancing~\cite{ananta,maglev}, 
network 
slicing~\cite{sherwood2010carving},
server migration~\cite{ton13mig},
switching~\cite{p4,smartnic}, or
even adjusting the network topology~\cite{projector}. 
The impact of distributed applications
on the communication network
is also well-documented in the literature~\cite{talk-about,kraken,dist-ml,singh2015jupiter,cisco-pro}.
Several empirical studies exploring the spatial and temporal 
locality in traffic patterns found evidence that these 
workloads are often \emph{sparse
and skewed}~\cite{alizadeh2010data,projector,Roy:2015,judd2015attaining}, introducing optimization opportunities.
E.g., studies of reconfigurable datacenter
networks~\cite{firefly,projector} 
have shown that for certain workloads,
a demand-aware datacenter network 
can achieve a performance similar to a demand-oblivious
datacenter network at 25-40\% lower cost~\cite{firefly,projector}. 

However, much less is known about the algorithmic challenges
underlying such workload-adaptive networked systems,
the focus of our paper. 
From an online algorithm perspective, 
our problem is related to reconfiguration problems
such as online page (resp.~file) migration~\cite{best-deterministic,black1989competitive}
as well as server migration~\cite{Bienkowski:2014:WVS:2591204.2591215} problems, 
$k$-server~\cite{k-server2} problems, or online metrical task systems~\cite{metrical-task-systems}.
In contrast to these problems, in our model,
requests do not appear somewhere in a graph or metric
space but \emph{between communication partners}.
From this perspective, our problem can also be seen as a ``distributed'' 
version of online paging problems~\cite{paging-mark,companion-caching,competitive-analysis,young-paging-soda} (and especially their variants \emph{with bypassing}~\cite{generalized-caching-optimal,caching-rejection-penalties})
 where access costs can be
avoided by moving items to a \emph{cache}: in our model,
access costs are avoided by collocating communication partners
on the same \emph{server} (a ``distributed cache'').

The static version of our problem, how to partition a 
graph, is a most 
fundamental
and well-explored problem in computer science~\cite{Vaquero:2013:APL:2523616.2525943}, with many applications, e.g.,
in community detection~\cite{abbe2018community}. The
balanced graph partitioning problem is 
related to minimum bisection 
problems~\cite{feige2002}, and
known to be hard even to 
approximate~\cite{andreev2006balanced}. 
The best approximation today is due to Krauthgamer~\cite{Krauthgamer2006}. 
In contrast, we in this paper are interested in a dynamic
version of the problem where the edges of the to-be-partitioned
graph are revealed over time, in an online manner.
Further, the offline problem of embedding workloads in a communication-efficient
manner has been studied in the context of the minimum linear arrangement
problem~\cite{mla} and the virtual network embedding problem~\cite{vnep}, 
however, without considering
the option of migrations.
In this regard, our paper features an interesting connection to the
\emph{itinerant list update} model~\cite{waoa18}, a kind of ``dynamic''
minimum linear arrangement problem 
which allows for reconfigurations
and, notably, considers pair-wise requests. However, 
communication is limited to a linear line and so far, only
non-trivial \emph{offline} solutions are known.

One of the applications of the problem we study is a distributed union find data
structure (see Section~\ref{sec:union-find}). Union find data structures have
been initially proposed in the centralized setting and efficient algorithms were
derived~\cite{galler64improved,tarjan84worst}. Later, parallel versions of union
find data structures were considered in a shared memory setting in which the
goal was to derive wait-free algorithms~\cite{anderson91wait}; also external
memory algorithms were considered~\cite{agarwal10efficient}. To the best of our
knowledge studies of union find data structures in a distributed memory setting
were only conducted experimentally, see (for example)
\cite{cybenko88practical,manne09scalable,patwary10experiments,patwary12multi}.

The second application we presented was as online $k$-way partitioning
(Section~\ref{sec:partition}). The $k$-way partitioning problem is known to be
\NP-hard as it constitutes a very simple scheduling
problem~\cite{garey79computers}. The problem has also been researched in
practice, see, e.g., \cite{korf09multi,ethan18optimal} and references therein.
We are not aware of literature studying the online version of the problem which
we have considered.

The paper most closely related to ours
is by Avin et al.~\cite{obr-original,avin19dynamic}
who studied a more general version of
the problem considered in our paper. 
In their model, request patterns can
change arbitrarily over time, and in particular,
do not have to follow a partition
and hence ``cannot be learned''. Indeed,
as we have shown in this paper, learning algorithms
can perform significantly better: 
in~\cite{obr-original}, it was shown that for constant $\ell$ 
any deterministic online algorithm must have a competitive
ratio of at least $\Omega(n)$ unless it can 
collocate \emph{all} nodes on a single server, while we have presented an
$O(\lg{n})$-competitive online algorithm.
Thus, our result is exponentially better than what can possibly be achieved in
the model of~\cite{obr-original}.

\section{Conclusion}
\label{sec:conclusion}

Motivated by the increasing resource allocation flexibilities
available in modern compute infrastructures,
we initiated the study of online algorithms
for adjusting the embedding of workloads according to
the specific communication patterns, to reduce
communication and moving costs. In particular, we presented
algorithms and derived upper and lower bounds on
their competitive ratio.

We believe that our work opens several interesting
questions for future research. In particular, 
it remains to close the gap between the upper and lower
bound of the competitive ratios derived in this paper.
Furthermore, while in this paper we assumed that there are $\ell$ ground truth
components of size $n/\ell$, it will be interesting to study more general
settings with smaller and larger components.

More generally, it will be interesting to consider algorithms which do not
collocate all communication partners. Also, studying collocation in specific
networks such as Clos networks, which are frequently encountered in datacenters,
would be intriguing.

\section*{Acknowledgments}
We are grateful to our shepherd Rachit Agarwal as well as the anonymous
reviewers whose insightful comments helped us improve the presentation of the
paper.

The research leading to these results has received funding from the European
Research Council under the European Community's Seventh Framework Programme
(FP7/2007-2013) / ERC grant agreement No.\ 340506.
Stefan Neumann gratefully acknowledges the financial support from the Doctoral
Programme ``Vienna Graduate School on Computational Optimization'' which is
funded by the Austrian Science Fund (FWF, project no.\ W1260-N35).

{\balance
  \bibliographystyle{plain} 
\bibliography{main}
}

\end{document}